\documentclass[lettersize,journal]{IEEEtran}
\usepackage{amsmath,amsfonts,amssymb}
\usepackage{amsthm,bbold}
\usepackage[frak=mma]{mathalpha}
\usepackage{empheq}
\usepackage{algorithm,algpseudocode}
\usepackage{array}
\usepackage[caption=false,font=normalsize,labelfont=sf,textfont=sf]{subfig}
\usepackage{textcomp}
\usepackage{tcolorbox}
\usepackage{stfloats}
\usepackage{url}
\usepackage{verbatim}
\usepackage{graphicx}
\usepackage{cite}
\usepackage{tabularx}%
\usepackage{graphicx,booktabs, colortbl}
\usepackage{multirow}
\usepackage{soul}
\usepackage{xcolor,soul}
\usepackage{enumitem}

\usepackage{mathtools}
\mathtoolsset{showonlyrefs}

\newtheorem{theorem}{Theorem}
\newtheorem{corollary}{Corollary}
\newtheorem{pro}{Proposition}
\newtheorem{definition}{Definition}

\newcommand*{\etaltwo}{\textit{et al}.\@ }

\newtheorem{lemma}{Lemma}

\usepackage{bm}

\begin{document}
\title{From Content Sharing to Opinion Shaping: \\ A Theoretical Framework for Cascade-Driven Opinion Dynamics}
%\title{From Sharing to Shaping: \\ Closed-form Analysis of Opinion Dynamics on Cascades}
%\title{Closed-Form Analysis of Cascade-Driven Opinion Dynamics in Social Networks}
%\title{A Theoretical Framework for Cascade-Driven Opinion Dynamics in Social Networks}
% other idea:
% Opinion Propagation Through Cascades: Theory, Approximations, and Community Effects

\author{Elisabetta Biondi, Chiara Boldrini, Andrea Passarella, Marco Conti
        % <-this % stops a space
\thanks{All authors are with the Institute of Informatics and Telematics (IIT) of the National Research Council (CNR), Italy. email: first.last@iit.cnr.it}% <-this % stops a space
% \thanks{The work of M. Conti and A. Passarella is partly supported by PNRR - M4C2 - Investimento 1.3, Partenariato Esteso PE00000013 - ``FAIR - Future Artificial Intelligence Research'' - Spoke 1 "Human-centered AI", funded by the European Commission under the NextGeneration EU programme. The work of C. Boldrini is partly supported by project SERICS (PE00000014) under the MUR National Recovery and Resilience Plan funded by the European Union - NextGenerationEU.}% <-this % stops a space
}



\maketitle

\begin{abstract}
Online social networks (OSNs) tightly couple information diffusion and opinion formation, yet most existing models treat these processes separately. Building on the Friedkin–Johnsen on Cascades (FJC) model, which integrates opinion updating within stochastic information cascades, we advance its analytical characterization and systematic validation. We derive an exact conditional layer-wise recursion that reveals the statistical coupling between diffusion realizations and opinion evolution, showing that a closed-form recursion in unconditional expectations is not available. To overcome this limitation, we introduce a mean-field decoupling approximation (MFA) that yields a tractable linear recursion for the expected opinion trajectory. We also develop scalable tree-based reductions that drastically reduce computational complexity while preserving structural diffusion properties, and validate the analytical recursion and its approximations through extensive simulations.
Using Stochastic Block Model  networks with controlled initial opinion configurations, we show that, unlike the classical Friedkin–Johnsen model, which typically leads to global homogenization, the FJC dynamics generate richer regimes of bounded assimilation and sustained polarization. Overall, this work bridges diffusion and influence modeling, providing a theoretical foundation for understanding how cascade-driven interactions reshape collective opinion landscapes in structured social networks.

\end{abstract}


\vspace{-10pt}
\section{Introduction}
\label{sec:intro}

\IEEEPARstart{O}{nline} social networks (OSNs) have profoundly reshaped how individuals access information, communicate, and form opinions. Platforms such as X (formerly Twitter), Facebook, Instagram, and TikTok have turned communication into a large-scale, real-time, and highly interactive process, transcending geographical and social boundaries. Nowadays, OSNs are used to share and express opinions, related to news\cite{shearer2021news,pentina2014information}, review products or services~\cite{pradhan2023social}, and even health advice~\cite{smith2008social}. 
For this reason, political actors, marketers, and even news organizations now complement traditional channels with posts, reshares, and user-generated content to enhance persuasion, visibility, or engagement~\cite{allcott2017social,akpinar2017valuable}. Empirical evidence shows that these processes are tightly intertwined: in a 24-hour snapshot of Twitter activity, more than half of all posts were retweets, while only one-fifth were original messages~\cite{pfeffer2023just}. 
Consequently, information dissemination and opinion formation are no longer separable processes: they co-occur and mutually reinforce each other within the same digital infrastructure. %As OSNs increasingly mediate access to information and public discourse, the boundaries between interpersonal communication, media exposure, and social influence have become more porous~\cite{pentina2014information}.

Traditionally,  research on opinion dynamics and information diffusion has largely developed along independent trajectories, where opinion dynamics models describe how individuals revise their beliefs under social influence, and diffusion models study how information, innovations, or behaviors propagate across networks. 
On the one hand, among classical opinion dynamic frameworks, the \emph{Friedkin-Johnsen (FJ)} model~\cite{friedkin1990social} remains a cornerstone. Belonging to the family of averaging models~\cite{degroot1974reaching,french1956formal}, it formalizes opinion evolution as an iterative weighted averaging process in which each agent updates its belief as a convex combination of its initial predisposition and the opinions of its neighbors. 
Due to its mathematical tractability and its empirical validation in small- to medium-scale social settings, the FJ model has gained considerable popularity, and has been used to explain collective wisdom~\cite{Das2013}, polarization~\cite{biondi2023dynamics}, filter bubbles~\cite{chitra2020analyzing}, and even negotiation processes~\cite{bernardo2021achieving}.
On the other hand, the dissemination of content is often modeled as an \emph{information cascade}~\cite{easley2010networks}, where each reshare depends on prior exposures and perceived relevance. These cascades act as vehicles through which both information and opinions travel, dynamically shaping collective beliefs.
%
%However, online communication introduces new forms of social exposure that differ substantially from offline deliberation. In OSNs, influence is shaped by algorithmic filtering, selective exposure, and asymmetric visibility of opinions~\cite{garimella2018quantifying,suarez2025cross}. Individuals encounter others’ views indirectly—through posts, comments, and reshares—rather than mutual conversation. Such asymmetric and large-scale exposure reinforces cognitive dissonance~\cite{festinger1957theory,tsfati2003media} and amplifies emotionally charged or polarizing content~\cite{vosoughi2018spread,delvicario2016misinformation}. The dissemination of this content is often modeled as an \emph{information cascade}~\cite{easley2010networks}, where each reshare depends on prior exposures and perceived relevance. These cascades act as vehicles through which both information and opinions travel, dynamically shaping collective beliefs.
%

Recently, there have been efforts to blur the lines of this disciplinary separation, where information cascades are also cascades of opinions, chains through which attitudes are amplified, reshaped, or contested as they spread~\cite{tu2022viral,brooks2024opinionreproductionnumberinfodemics}. The Friedkin–Johnsen on Cascades (FJC) model~\cite{biondi2026cascade} is an analytical and computational framework that explicitly links information diffusion and opinion dynamics in online networks, extending the classical FJ formulation by embedding opinion updates within the temporal and structural process of an information cascade. As content spreads, each reshare transmits not only the post but also the sharer’s current opinion, allowing beliefs to evolve through exposure. Thus, diffusion becomes a mechanism of influence: opinion change propagates along the same paths as information.
In~\cite{biondi2026cascade}, we analyzed the FJC model through simulations, showing that cascade-mediated exposure can drastically alter classical FJ dynamics by amplifying the influence of central nodes and reshaping polarization patterns observed in the classical FJ.  
However, a rigorous analytical characterization of cascade-driven opinion trajectories has remained open.
%However, the polarization metrics used in that analysis, drawn from existing opinion dynamics literature~\cite{biondi2023dynamics}, typically consider only the distribution of opinion values, neglecting how these opinions are arranged within the network. However, one of the main characteristics in social networks is the \emph{homophily}, the tendency of similar individuals to connect, which critically affects the emergence and persistence of group divisions and polarization~\cite{mcpherson2001birds}. Intuitively, opinion values themselves cannot take into account this dimension, as identical opinion distributions can yield very different outcomes depending on how opinions are positioned across the network structure. Since homophily is not inherently encoded in the FJ model (unlike, for example, in the bounded confidence model~\cite{hegselmann2002opinion}), this limitation becomes particularly relevant, as the FJ framework has been criticized for its limited ability to reproduce clustered or polarized opinion structures~\cite{bernardo2024bounded}.

%In this work, we address these limitations in two main directions.  
In this paper, we provide such a characterization. First, we derive an exact conditional layer-wise recursion that explicitly reveals the statistical coupling between diffusion realizations and opinion evolution. This result shows that, in general, a closed recursion involving only unconditional expectations cannot be obtained. Second, to achieve analytical tractability, we introduce a mean-field decoupling approximation (MFA) that neglects the diffusion–opinion correlation at the expectation level. This approximation yields a closed linear recursion for the expected opinion trajectory, enabling scalable computation and theoretical analysis. Importantly, we empirically validate the accuracy of this analytical recursion through extensive simulations of the full stochastic process. 
We further develop scalable tree-based reductions that drastically lower computational complexity while preserving the structural properties of diffusion, and we validate their performance against both the analytical model and simulations. 

%In this work, we make a step forward in two main directions.
%First, we provide the analytical formulation of the FJC model, together with theoretical results and efficient numerical approximations that enable scalable computation of opinion trajectories and extensive validation through simulations. This analytical perspective is particularly important, as it provides direct insight into how influence propagates through the network, enabling the identification of the most influential nodes and pathways. Moreover, it lays the groundwork for the development of analytical tools to design and optimize targeted interventions—either at the network level or along the propagation flow—to mitigate the spread of toxic or harmful behaviors. \\
Finally, we extend the analysis of FJC and study its behavior in networks composed of communities, where initial opinion configurations may or may not align with the community structure. In other words, the network topology itself is not assumed to be homophilic; instead, homophily emerges (or does not) depending on how initial opinions are assigned.  %Specifically, we examined how local exposure patterns within cascades translate into large-scale opinion shifts when opinions are clustered by community. 
To explore these scenarios, we conduct experiments on synthetic social structures generated through \emph{Stochastic Block Models (SBMs)} with two communities and controlled initial opinion configurations, conducting a series of \emph{what-if} analyses to probe the conditions under which cascades lead to homogenization, bounded assimilation, or radical divergence. 
%Second, we extend the analysis of FJC to incorporate \emph{homophilic polarization}, examining how local exposure patterns within cascades translate into large-scale opinion shifts when opinions are clustered by community. Specifically, we conduct experiments on synthetic social structures generated through \emph{Stochastic Block Models (SBMs)} with two communities and controlled initial opinion configurations, performing a set of \emph{what-if} analyses to explore the conditions under which cascades lead to homogenization, bounded assimilation, or radical divergence. 
%
Our findings show that the averaging property of the FJ model, which by design brings opinions closer together, is dramatically altered in the presence of cascades. Firstly, cascades preserve larger intergroup opinion distances (\emph{polarization}). Secondly, even in scenarios where the classical FJ model leads to global \emph{homogenization}, the FJC model exhibits \emph{bounded assimilation}, whereby groups of initially aligned opinions maintain their original structure. These results suggest that the FJC model captures more realistic online behavior, in which initially aligned opinion groups tend to preserve their structure. 
%Our findings show that while the classical FJ model typically drives networks toward global \emph{homogenization} through iterative averaging, the FJC model captures more realistic online behaviors. This gives rise to regimes of \emph{bounded assimilation}, in which opinions converge within initially aligned groups, and \emph{polarization}, where intergroup exposure intensifies divergence. 
Overall, the proposed framework bridges a long-standing divide between diffusion and influence modeling, offering a theoretical foundation for understanding how cascade-driven interactions reshape collective opinion landscapes in networks with structured communities, independent of whether those communities initially exhibit opinion homophily.

\vspace{-5pt}
\section{Related Work}
\label{sec:relwork}
Research on FJ-based opinion dynamics has evolved along two main directions.  
The first focuses on optimization and intervention strategies~\cite{Gionis2013,Bindel2015,Tu2023,Sun2023,abebe2021opinion}, where the goal is to steer collective opinions through minimal or targeted influence actions.  
The second extends the classical model to incorporate richer social mechanisms, such as antagonism~\cite{bilo2022general,Razaq2023}, temporal interactions~\cite{Bhalla2023}, and algorithmic filtering~\cite{chitra2019understanding}.  

Several studies explicitly incorporate \emph{homophily}, the tendency of similar individuals to connect~\cite{mcpherson2001birds}.  
In~\cite{disaro2024balancing}, the social network is time-varying and homophily-dependent, allowing users to rewire their connections to increase opinion agreement.
In contrast, we focus on static networks, assuming that social ties reflect relatively stable affective or identity-based relationships.
Dandekar~\etaltwo~\cite {dandekar2013biased} incorporate biased assimilation, whereby agents are more strongly influenced by opinions aligned with their own. In~\cite{chitra2019understanding}, Chitra \etaltwo introduce an algorithmic mediator that modifies network edges to reduce exposure to opposing views, thus modeling the emergence of filter bubbles.
Both approaches examine homophily within SBM structures, as we do; however, their polarization measures capture only local disagreement\cite{dandekar2013biased} or global divergence\cite{chitra2019understanding}.
Our approach extends beyond these metrics by analyzing polarization in both opinion space and network topology, revealing how cascades can sustain, amplify, or mitigate polarization.
This enables us to link micro-level diffusion to macro-level structural polarization, providing insights into how opinion diversity and network segregation jointly shape collective outcomes in online social systems.

In parallel, research on \emph{information diffusion}, notably the Independent Cascade (IC) model~\cite{kempe2003maximizing,guille2013information,zhou2021survey}, has explored the spread of content and behaviors.   While the diffusion literature primarily focuses on the reach and speed of propagation, less attention has been paid to how these diffusion processes co-evolve with opinion formation.
Few works couple diffusion with opinion formation~\cite{tu2022viral,brooks2024opinionreproductionnumberinfodemics}. Both couple diffusion and opinion dynamics only loosely: in the former, content affects opinions once through initial bias shifts, while in the latter, opinion adoption depends solely on similarity to the source, modeling virality rather than belief revision. 
Our prior study~\cite{biondi2026cascade} introduced the FJC model, integrating the FJ influence mechanism into the temporal structure of diffusion cascades, linking opinion revision directly to exposure events.  
In the present study, we build on that framework analytically, deriving theoretical results and approximate solutions to estimate steady-state opinions under cascade-driven exposure.
%This bridges the gap between simulation-based exploration and theoretical characterization, allowing efficient evaluation of opinion trajectories in large-scale networks and a systematic assessment of how the temporal and structural properties of cascades shape opinion evolution.

\vspace{-5pt}
\section{Preliminaries and Model Overview}
\label{sec:preliminaries}

%This section introduces the notation and theoretical background required throughout the paper.

%\vspace{-10pt}
\subsection{Network Representation}
\label{sec:social_network}
The social network is modeled as a directed graph $\mathcal{G}=(\mathcal{V},\mathcal{E})$, 
with  $\mathcal{V}=\{1,\ldots,n\}$ representing the set of users and 
$\mathcal{E}\subseteq \mathcal{V}\times\mathcal{V}$ the set of directed ties.
An edge $(j,i)\in\mathcal{E}$ indicates that information flows from user $j$ to user $i$, namely,  $i$ can observe and be influenced by $j$.  For each node $i$, let %$\mathcal{N}^+(i)=\{j:(i,j)\in\mathcal{E}\}$ be the users influenced by $i$, and 
$\mathcal{N}^{in}(i)=\{j:(j,i)\in\mathcal{E}\}$ be the set of users influencing $i$. %; these coincide in undirected graphs.  
Interactions between each pair of nodes $(i,j)$ typically quantify the strength of their relationships, denoted as $\hat{w}_{ij}$ and are collected into a symmetric matrix $\hat{W}=(\hat{w}_{ij})\in\mathbb{R}^{n\times n}$, referred to as \emph{social matrix}.  The corresponding \emph{influence matrix} $W=(w_{ij})$ is row-stochastic, where $w_{ij}\in[0,1]$ denotes the influence of $j$ on $i$ and is typically obtained by normalizing $\hat{W}$,  
$w_{ij}=\hat{w}_{ij}/\sum_k \hat{w}_{ik}$.  
Although $\hat{W}$ is symmetric, $W$ may be asymmetric, reflecting unequal influence weights.

\vspace{-10pt}
\subsection{Information cascade}
\label{sec:information_cascade} 

To describe information propagation, we employ the \emph{independent cascade (IC)} model~\cite{kempe2003maximizing}. In this model, when a node $v$ in the network posts content, it transitions from an  \emph{inactive} state to an  \emph{active} state, i.e. it \emph{activates}. Once this transition occurs, the node not only spreads the content to its peers, but also has the chance to activate those neighbors who are still in the inactive state. This activation process is represented by the probability \( p_{vw} \), where \( w \) is a neighbor of \( v \) and the pair \( (v,w) \) belongs to the set of edges \( \mathcal{E} \).
%Once activated, a node has only one chance to influence its neighbors, and the diffusion process continues until no further activations occur. 
%This stochastic framework effectively captures online dissemination phenomena such as retweeting or reposting behavior, in which exposure and reshare occur at most once per post.

%Our model conceptualizes the act of posting social content following the independent cascade model~\cite{kempe2003maximizing}. In the independent cascade model, each node can be \emph{active} or \emph{inactive}. An initial set of active nodes $v$ have a single chance to activate their neighbors $w$, and each activation is successful with probability $p_{vw}$. The nodes $w$ that are successfully activated, become active at the next step and try to activate their neighbors as well.  The process continues until no other nodes can be activated.

\vspace{-10pt}
\subsection{Friedkin-Johnsen  model}
\label{sec:fj}
The Friedkin-Johnsen (FJ) model~\cite{friedkin1990social} formalizes how opinions evolve under social influence and individual conviction.  
Each agent $i$ holds an opinion $x_i(t)\in[-1,1]$ at discrete time $t$, shaped by its initial belief $u_i$ and a susceptibility parameter $\lambda_i\in[0,1]$ controlling openness to influence.  
Agents with $\lambda_i = 0$ are stubborn, while those with $\lambda_i= 1$ are fully adaptive.  
Let $\Lambda=\mathrm{diag}(\lambda_1,\ldots,\lambda_n)$ be the diagonal susceptibility matrix and $W$ be the row-stochastic influence matrix, described in Section~\ref{sec:social_network}, the opinion update rule for a node $i$ is given by:
\vspace{-5pt}
\begin{equation}
    x_i(t+1) = \lambda_i \sum_{j=1}^{n} w_{ij} x_j(t) + (1 - \lambda_i) u_i.
    \label{eq:fj_update} \vspace{-7pt}
\end{equation}
This equation describes a linear convex combination between  an individual's current belief (anchored to $u_i$) and the weighted average opinion of its influencers. The parameter $\lambda_i$ therefore controls the relative influence of the social context.
The vector form of the opinion update is given by $\bm{x}(t+1) = \Lambda W \bm{x}(t) + (I - \Lambda)\bm{u}$. 
% :
% \vspace{-5pt}
% \begin{equation}
%     \bm{x}(t+1) = \Lambda W \bm{x}(t) + (I - \Lambda)\bm{u}.
%     \label{eq:fj_vector}
% \end{equation}
%
Under mild conditions, namely, a strongly connected network and at least one partially anchored node (\(\lambda_i<1\))\footnote{Please note that these assumptions are not restrictive. If connectivity is not achieved,  the problem can be defined using the connected components. Instead, when $\Lambda=I$, the FJ model simplifies to the French-DeGroot model~\cite{degroot1974reaching}, which has been widely studied.}, the process converges to a unique equilibrium~\cite{proskurnikov2017tutorial}:
\vspace{-5pt}
\begin{equation}
    \bm{z} = (I - \Lambda W)^{-1} (I - \Lambda)\bm{u}.
    \label{eq:fj_equilibrium} \vspace{-7pt}
\end{equation}

%The FJ model thus captures a continuous compromise between self-confidence and social conformity. The term $\Lambda W$ encodes the propagation of influence along network ties, while $(I-\Lambda)$ enforces persistence of individuality.

The standard FJ model assumes \emph{synchronous} updates, whereas real social exchanges are inherently asynchronous.  
To capture this, Frasca \emph{et al.}~\cite{frasca2013gossips} proposed a stochastic variant where, at each iteration, a random edge $(j,i)$ is activated and only node $i$ updates its opinion, as:
\begin{equation}
    x_i(t+1) = h_i\big[(1-\gamma_{ij})x_i(t) + \gamma_{ij}x_j(t)\big] + (1-h_i)u_i,
    \label{eq:fj_asyn_update_i}
\end{equation}
In this asynchronous formulation, the matrices $H = \{h_{ij}\}$ and $\Gamma = \{\gamma_{ij}\}$ replace the $\Lambda$ and $W$ matrices of the synchronous model, but they are not identical to them. To ensure that the asynchronous dynamics in Eq.~\eqref{eq:fj_asyn_update_i} converge to the same steady state as the synchronous model in Eq.~\eqref{eq:fj_update}, one must account for the probability that node $i$ is selected for updating, and thus for the node degree $d_i$. 
Accordingly, whenever $H$ and $\Gamma$ are defined in terms of $\Lambda$ and $W$ as follows:
\begin{align}
    h_i &= 1 - \frac{1-\lambda_i}{d_i}, \label{eq:hi}\\
    \gamma_{ij} &=
    \begin{cases}
        \dfrac{d_i(1-h_i) + h_i - (1-\lambda_i w_{ii})}{h_i}, & i=j, \\[6pt]
        \dfrac{\lambda_i w_{ij}}{h_i}, & i\neq j.
    \end{cases}\label{eq:gammaij}
\end{align}
The asynchronous process converges \emph{ergodically} to the same equilibrium $\bm{z}$ of the standard FJ model given by Eq.~\eqref{eq:fj_equilibrium}
~\cite{frasca2013gossips}. The asynchronous version is the one used in the FJC model. 
%although instantaneous opinion fluctuations, the time averages of the asynchronous process converge to the same equilibrium $\bm{z}$ given by~\eqref{eq:fj_equilibrium}, a property known as \emph{ergodic convergence}~\cite{frasca2013gossips}.

\vspace{-7pt}
\subsection{The Friedkin-Johnsen on Cascades Model}
\label{sec:FJC}

The FJC model integrates opinion dynamics with the stochastic nature of online information diffusion, coupling the FJ model with an IC process. In this framework, a cascade determines when and between which nodes interactions occur, while the FJ update rule specifies how each exposure affects opinions. Thus, the model preserves the interpretability of the classical FJ formulation, but embeds it within the short, burst-like exposures typical of online platforms. In the following, we briefly recall the operational mechanics of the FJC model as introduced in~\cite{biondi2026cascade}, which forms the basis for the analytical derivation in the following section, where the contributions of the present paper begin. 
%we will briefly describe the parameters and mechanisms of the model that are necessary for the theorettical derication provided in Section~\ref{sec:modelling}. For futher details, please refer to~\cite{biondi2026cascade}.

The model considers posts spreading one at a time. When a user $r$ initiates a post, the diffusion process unfolds as a tree rooted at $r$. In fact, as for the largest part of opinion dynamics models, in FJC a node is exposed only to social activities of nodes to whom it is linked and does not incorporate any recommendation strategy. The depth of this tree is equal to its eccentricity $\epsilon^{(r)}$, defined as the maximum graph distance from $r$ to any other node; consequently, $\epsilon^{(r)} \le L$, where $L$ is the diameter of the network.  We denote this deterministic reachability tree by
$\mathcal{O}^{(r)}=\left(\mathcal{O}^{(r)}_l\right)_{l=0}^{\epsilon^{(r)}}$, where each layer  $\mathcal{O}^{(r)}_l$ contains all nodes at distance $l$ from the root $r$.  %Because the network diameter is $L$,  it holds that $\epsilon^{(r)} \le L$. 
For convenience, the tree sequence $\mathcal{O}^{(r)}$ is extended to exactly $L$ layers by defining $\mathcal{O}^{(r)}_l = \emptyset$ for $l > \epsilon^{(r)}$, thus $\mathcal{O}^{(r)}=\big(\mathcal{O}^{(r)}_0,\dots,\mathcal{O}^{(r)}_L\big)$. This convention simplifies the notation without altering the dynamics.

The actual cascade does not activate all nodes in $\mathcal{O}^{(r)}$. Instead, each node $i$ reshares the post independently with probability $\theta_i$, generating a random subtree
$\mathcal{T}^{(r,\theta)} = \big(\mathcal{T}^{(r,\theta)}_0,\dots,\mathcal{T}^{(r,\theta)}_L\big)$
of the baseline layering, i.e.  $\mathcal{T}^{(r,\theta)}_l\subseteq\mathcal{O}^{(r)}_l$ . The cascade $\mathcal{C}^{(r)}$ is the collection of nodes that are involved in $\mathcal{T}^{(r,\theta)}$, i.e. $\mathcal{C}^{(r)} = \bigcup_l\mathcal{T}^{(r,\theta)}_l$. Figure~\ref{fig:example} illustrates a simple network with the corresponding tree.

As the cascade evolves, nodes transition through the states \emph{unaware}, \emph{aware}, and \emph{spreader}. Upon receiving the post from one or more predecessors in layer $l-1$, a node $i$ becomes aware and updates its opinion. Finally, with probability $\theta_{i}$, it becomes a spreader and reshares the post, enabling nodes in layer $l+1$ to become aware in their turn, provided $l < L$. 
%Opinion updates occur exactly at moments of exposure. When a node $i$ receives the post from a set of predecessors $\varphi(i)$ in the previous layer, the model assumes that it processes all incoming opinions at the same time, reflecting how users typically encounter multiple comments or reposts at once. 
Opinion updates occur exactly at moments of exposure. When a node $i$ receives the post from predecessors $\varphi(i)$, it observes the post together with the opinions attached to all their reshares. The model assumes that $i$ processes this set of opinions as a whole and performs a single, aggregated update accordingly. Its opinion update follows the asynchronous FJ model (Eq.~\eqref{eq:fj_asyn_update_i}), and it is averaged over the influence of all its predecessors. Specifically, the opinion is updated as follows:
\begin{align} %\vspace{-10pt}
    x_{i}(t+1) &= \frac{1}{|\varphi(i)|} 
    \Bigg\{ \sum_{j\in\varphi(i)} h_{i} \big[(1-\gamma_{ij}) x_{i}(t)
    +\gamma_{ij} x_{j}(t)\big] \nonumber \\
    &\qquad + (1-h_{i})u_i \Bigg\},
    \label{eq:fjc_update} %\vspace{-10pt}
\end{align}
where $h_i$ and $\gamma_{ij}$ are defined as in the asynchronous FJ model, and $\varphi(i)$ is the set of predecessors resharing the post. Please note that Eq.~\eqref{eq:fjc_update} incorporates the concept that a node is only exposed to opinions from nodes that are engaged in a neighborhood relationship.
The opinion update process unfolds over a sequence of posts \(p_1,\dots,p_N\), rooted at nodes \(r_1,\dots,r_N\). The first cascade starts from the opinion vector $\bm{u}$; each subsequent cascade then operates on the opinion state resulting from all preceding cascades.%The first cascade begins from the initial opinion vector $\mathbf{u}$, whereas each subsequent cascade starts from the final opinions generated by the previous post. 

\begin{figure}[t]
    \centering
      \subfloat[Network]{
          \includegraphics[width=0.15\textwidth]{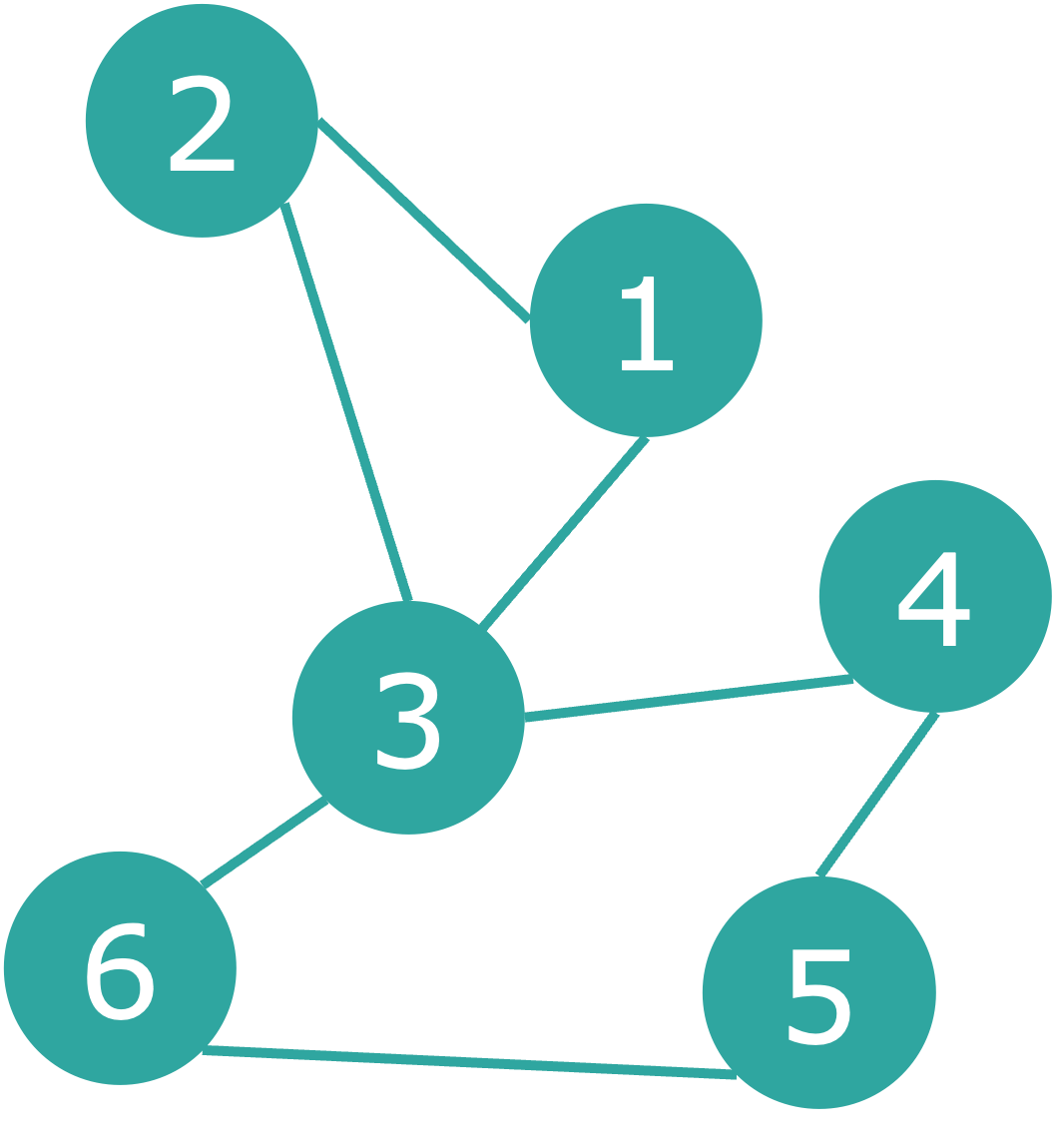}
      }\hfill
      \subfloat[$\mathcal{T}^{(1,\theta)}$ ]{
          \includegraphics[width=0.2\textwidth]{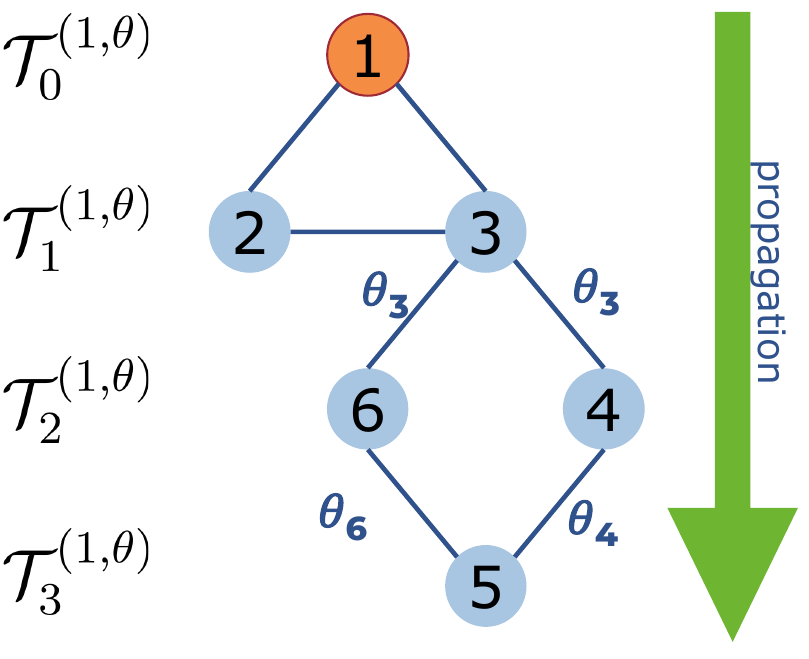}
      }\
    \caption{Example of a network (a) and the corresponding stochastic cascade tree (b) initiated by node $r=1$. Nodes $2$ and $3$ deterministically update their opinions, while the others participate probabilistically according to $\theta$.}
    \label{fig:example}
    \vspace{-10pt}
\end{figure}

% \vspace{-5pt}
% Formally, the cascade $\mathcal{C}^{(r)}$ generated by the post $p$ from node $r$ is defined as the set of nodes in $\mathcal{T}^{(r,\theta)}$. Nodes in layer~$1$ always belong to the cascade, as they receive the post directly from the root $r$. Nodes in layers $l\ge 2$ participate only if at least one predecessor in layer $l-1$ reshared the post. Let us denote by $S_i$ the Bernoulli random variable indicating whether node $i\neq r$ reshares the post ($S_i = 1$ if reshared, $0$ otherwise). Then, the cascade tree can be constructed recursively:
% \begin{equation}
%     \left\{
%     \begin{array}{l}
%         \mathcal{T}^{(r,\theta)}_{0} = \{r\},\\[2pt]
%         \mathcal{T}^{(r,\theta)}_{1} = \mathcal{N}^-(r),\\[2pt]
%         \mathcal{T}^{(r,\theta)}_{l+1} = 
%             \left(\bigcup\limits_{\{i \in \mathcal{T}^{(r,\theta)}_l : S_i = 1\}} \mathcal{N}^-(i) \right) 
%             \setminus  \mathcal{T}^{(r,\theta)}_l.
%     \end{array}
%     \right.
% \end{equation}
%The stochasticity nature of the problem implies that a node belongs to the cascade not deterministically, except for $r$ (the root of the cascade) and its neighbors that receive the post. 

Summarizing, the model is built on the following three assumptions:
\begin{enumerate}[label = (A\arabic*)]
    \item Opinion-neutral content: shared content does not carry an intrinsic opinion; opinions are conveyed only by the individuals sharing it. 
    \item Tree-structured exposure: a post propagates from the seed in a tree-like manner across successive neighbor layers, with each node influenced only upon first exposure. 
    \item Sequential cascade propagation: only one post propagates at a time; simultaneous cascades are not considered. 
\end{enumerate}
While assumptions (A2)-(A3) are introduced to ensure analytical tractability and may appear restrictive, we empirically validated them in~\cite{biondi2026cascade} using real Twitter cascade data. The results showed that, despite these simplifications, the model achieves high predictive accuracy in capturing observed opinion dynamics.
% Under these conditions, the model achieves strong predictive accuracy in reproducing observed opinion dynamics.

\vspace{-7pt}
\section{Analytical Derivation of Opinion Trajectories}
\label{sec:modelling}

The FJC framework described in Section~\ref{sec:FJC} establishes the operational rules by which opinion updates are embedded within stochastic information cascades. We now advance the framework from an algorithmic description to a rigorous analytical formulation.  In this section, we develop a theoretical formulation of the FJC dynamics. We begin by formalizing the temporal indexing and layer-based progression of the cascade (Section~\ref{sec:temporal_indexing}), which establishes the domain for the matrix representation of the opinion updates (Section~\ref{sec:matrix_rep}). We then quantify the probabilistic structure of cascade layers (Section~\ref{sec:probabilistic_structure}). These components are then combined to establish an exact conditional recursion that captures the coupling between diffusion and opinion dynamics, and to introduce mean-field and structural approximations that provide tractable predictions for large-scale networks (Sections~\ref{sec:fjc_solution} and \ref{sec:approximations}).

% In this section, we set up the notation and the preliminary results that will be useful for the model derivation in Section~\ref{sec:fjc_solution}.

\vspace{-5pt}
\subsection{Sequential Opinion Updates}
\label{sec:temporal_indexing}

Opinion updates in the FJC model (Section~\ref{sec:fj}) occur sequentially, following the order induced by the cascade realization. For each post $p_n$, nodes update their opinions layer by layer along the cascade tree, starting from the root and proceeding outward. 
Rather than tracking each individual update time, we exploit this ordering to define aggregated layer-wise states. Specifically, we denote by $\bm{y}^{(n)}_l$ the opinion vector after all nodes in layer $l$ of the $n$-th cascade have updated, with $\bm{y}^{(n)}_0$ representing the opinion vector at the beginning of the cascade. 
This construction induces a well-defined sequential update process, in which each node update depends on the current system state resulting from all prior updates. As a result, the FJC dynamics can be equivalently represented as a sequence of linear state transitions, which we formalize in the following section.

\vspace{-7pt}
\subsection{Matrix Representation of the Update Process}
\label{sec:matrix_rep}

With the sequential update structure established in the previous section, we now formalize the opinion revision mechanism at each step of the cascade. Each update corresponds to a single node revising its opinion based on the set of predecessors from which it receives the post.

Let $\bm{x}(t)$ denote the current opinion vector. When a node $i$ updates its opinion, the latter is revised according to Eq.~\eqref{eq:fjc_update}, while all other components remain unchanged. This defines a local transformation of the state, which can be expressed in matrix form as follows.

\begin{lemma}[Matrix form of the opinion update]\label{lemma:fjc_matrix} 
Let $\bm{x}(t)$ be the current opinion vector, and let node $i$ update its opinion based on the set of nodes in the previous layer that reshared the post, i.e. its predecessors $\varphi(i)$. Then the updated opinion vector $\bm{x}(t+1)$ can be written as
    \begin{equation}
        \bm{x}(t+1) = \frac{1}{|\varphi(i)|}\sum_{j\in \varphi(i)}A_{ij}\bm{x} (t)+B_i\bm{x}(t)+C_i\bm{u},
        \label{eq:fjc_matrix}
    \end{equation}
where the operators $A_{ij}$, $B_i$, and $C_i$ are defined as:
    \begin{align*}
        A_{ij}&=He_ie_i^T\left[I-\gamma_{ij}(e_ie_j^T-e_ie_i^T)\right]\\
        B_i&=I-e_ie_i^T\\
        C_i&=e_ie_i^T(I-H),
    \end{align*}
and $e_i$ is the $i$-th canonical basis vector.
\end{lemma}

\begin{proof} Matrix $B_i$ is the identity matrix, except for the $(i,i)$-th entry, which is zero. Thus, the second term of Eq.~\eqref{eq:fjc_matrix} gives the opinion vectors at step $t$ except for the  $i$-th entry equal to zero. On the other hand, for all $j$, matrices $A_{ij}$ and $C_i$ are zero in all rows except the $i$-th and therefore account for the updating of $i$'s opinion only. An easy algebraic computation shows that this update is exactly equal to that of Eq.~\eqref{eq:fjc_update}. Taking all the above considerations together, we obtain the thesis.
\end{proof}

Lemma~\ref{lemma:fjc_matrix} provides a matrix representation of a single update. We now extend this formulation to an entire layer of the cascade. 
Let $\mathcal{T}_l^{(r_n,\theta)}$ denote the set of nodes in layer $l$ of the realized cascade. Starting from the opinion vector $\bm{y}_{l-1}^{(n)}$, the final state $\bm{y}_l^{(n)}$ is obtained by applying sequentially the updates of all nodes in $\mathcal{T}_l^{(r_n,\theta)}$, in the order induced by the cascade. 
Since nodes in the same layer are distinct, each update affects a different component of the state. As a consequence, the sequential application of updates can be aggregated into a single expression evaluated at the common input $\bm{y}_{l-1}^{(n)}$.

\begin{corollary}[Opinions at the end of layer $l$ as a function of opinions at the end of layer $l-1$]\label{cor:end_opinion_layer}
    The final opinion at the end of layer $l$ of the cascade $\mathcal{C}$  can be expressed in matrix form as:
    \vspace{-3pt}
    \begin{align}
        \bm{y}_l^{(n)} &=\sum\limits_{i_m\in\mathcal{T}_l^{(r_n,\theta)}} \frac{1}{|\varphi(i_m)|} \sum_{j_m\in\varphi(i_m)}A_{i_m j_m} \bm{y}_{l-1}^{(n)} +\nonumber\\
        &\qquad+\tilde{B}_{\mathcal{T}_l^{(r_n,\theta)}}\bm{y}_{l-1}^{(n)}  + \sum\limits_{i_m\in\mathcal{T}_l^{(r_n,\theta)}}  C_{i_m}\bm{u}.
        \label{eq:end_opinion_layer}
    \end{align}
    \vspace{-3pt}
    for $l\geq 1$, where $\tilde{B}_{\mathcal{T}_l^{(r_n,\theta)}} = I - \sum\limits_{k \in \mathcal{T}_l^{(r_n,\theta)}} e_ke_k^T$.
    %Consequently, the final opinion at the end of the cascade $\bm{z}^{(n)}$ is equal to $\bm{y}_L^{(n)}$.
\end{corollary}

\begin{proof}

Let $\mathcal{T}_l^{(r_n,\theta)} = \{i_1, \dots, i_F\}$ be the set of nodes in layer $l$, ordered according to the sequence in which they update. Starting from $\bm{y}_{l-1}^{(n)}$, the final state $\bm{y}_l^{(n)}$ is obtained by applying the updates of nodes $i_1, \dots, i_F$ sequentially. 
By Lemma~\ref{lemma:fjc_matrix}, each update modifies only the component associated with the updating node. Since the nodes in $\mathcal{T}_l^{(r_n,\theta)}$ are distinct, these updates act on disjoint components of the state and therefore do not interfere. This follows from the fact that, for $i \neq k$, the corresponding update operators act on disjoint coordinates and therefore commute. As a result, their contributions can be aggregated into a single expression evaluated at $\bm{y}_{l-1}^{(n)}$. 
The first term accounts for the influence of predecessors, the second preserves the components of nodes that do not update, and the third incorporates the contribution of the initial opinions. Combining these terms yields the stated expression. 
\end{proof}

\vspace{-13pt}
\subsection{Probabilistic Structure of the Cascade}
\label{sec:probabilistic_structure}

The matrix formulation in Section~\ref{sec:matrix_rep} describes the opinion update for a specific, deterministic realization of the cascade. However, since the cascade is stochastic, the set of active predecessors $\varphi(i)$ is random. To compute expectations, we characterize both the probability that a node belongs to the cascade and the probability of each predecessor-set configuration. %However, the FJC model is inherently stochastic: the actual propagation tree $\mathcal{T}^{(r,\theta)}$ is a random subgraph of the baseline reachability tree $\mathcal{O}^{(r)}$, determined by the resharing probability $\theta$. Consequently, the set of active predecessors $\varphi(i)$ influencing a node $i$ is a random variable. To compute the expected opinion trajectories, we must quantify the likelihood of these structural configurations, which is our goal for this section. 
%In this section, we derive the probability $\delta_i$ that a generic node $i$ participates in the cascade, as well as the probability distribution over the possible subsets of active predecessors that trigger an opinion update.
% In the rest of this section, we will make some simple derivations about the nodes being part of a cascade, which will depend on the resharing probability of nodes. 
%Note that, in order to find a closed form, we make a simplifying assumption that the resharing probability is uniform across all pairs, i.e. $\theta_{i}=\theta$ for all $i\in \mathcal{V}$. We will leave as future work to relax this assumption.

We denote by $\delta_i$ the probability of a node $i$ being part of the cascade $\mathcal{C}^{(r)}$. For what we have already said, $\delta_i=1$ for nodes $i$ in layers~0 and~1. To generalize this probability to the other layers, we introduce the notion of potential predecessors.

\begin{definition}
    The \emph{parents} $\phi(i)$ of a node $i\in\mathcal{O}^{(r)}_{l}$ are the neighbors of $i$ belonging to the previous layer of the deterministic reachability tree, i.e.:
    \begin{equation}
        \phi(i)=\mathcal{N}^{in}(i)\cap\mathcal{O}^{(r)}_{l-1}.
    \end{equation}
\end{definition}

\noindent We can partition $\phi(i)$ into three disjoint subsets as follows:
\begin{equation}
\phi(i)=\varphi(i)\cup\xi(i)\cup\psi(i).
\label{eq:partition_pred}
\end{equation}
where $\varphi(i)$ collects the predecessors (nodes in the cascade who reshare),  
$\xi(i)$ the nodes in the cascade but not resharing, and  
$\psi(i)$ the nodes not in the cascade. Figure~\ref{fig:predecessors} illustrates these relations.  

Let $S_i$ be a Bernoulli random variable indicating whether node $i\neq r$ reshares the post ($S_i = 1$ if reshared, $0$ otherwise). Then node $i$ belongs to the cascade if and only if at least one of its parents belongs to the cascade and reshares the post, i.e., if the set below is nonempty:
\begin{equation}
    \varphi(i)=\left\{j\in \phi(i): j\in\mathcal{T}^{(r,\theta)}_{l-1} \wedge S_{j}=1  \right\}.
    %\varphi(i)=\left\{j\in \phi(i): j\in\mathcal{T}^{(r,\theta)}_{l-1} \wedge S_{j}=1  \right\}
\end{equation}
% is nonempty.

\begin{pro}[Probability of a node being part of the cascade]\label{pro:deltai}
    Consider a node $i$ in the tree $\mathcal{O}^{(r_n)}$ and let us denote by $\delta_i$ the probability $P(\{i \in \mathcal{C}^{(r_n)}\})$ of node $i$ being part of the cascade. % and let us indicate with  $\delta_i$ the probability $P(\{i\in\mathcal{C}\})$ of node $i$ being part of the cascade and with $p=\theta$ and $q=1-\theta$, w
    We can derive $\delta_i$ with the following equation:
    \begin{equation}
        \delta_i = \left\{
            \begin{array}{ll}
            1 & \qquad\text{if $i\in \mathcal{O}^{(r_n)}_1$}\\
            1-\prod_{a\in\phi(i)}(1-\theta_a) &\qquad \text{if $i\in \mathcal{O}^{(r_n)}_2$}\\
            1-\sum\limits_{\mathcal{B}\in\mathcal{P}\left(\phi(i)\right)} 
            \Bigg(\prod\limits_{b\in\mathcal{B}}(1-\delta_b)&\prod\limits_{d\in\phi(i)\setminus\mathcal{B}}\delta_d(1-\theta_d)\Bigg)\\
            & \qquad \text{if $i\in \mathcal{O}^{(r_n)}_{l>2}$}\\
            % \vspace{1mm}
            \end{array}
        \right.
        \label{eq:delta_i}
    \end{equation}
    where $\mathcal{P}(\phi(i))$ indicates the power set of $\phi(i)$, and $\theta_a$ indicates the resharing probability for a node $a$. 
\end{pro}

\begin{proof}

A node $i$ belongs to the cascade if and only if at least one of its parents both belongs to the cascade and reshares the post, i.e., if $\varphi(i) \neq \emptyset.$\\
\textsc{Layer 1.} 
Nodes in layer~1 receive the post directly from the root $r_n$. Therefore, $\phi(i) = \{r_n\}$ deterministically, and hence $\delta_i = 1$.\\
\textsc{Layer 2.} 
For nodes in layer~2, all parents belong to layer~1 and thus are always in the cascade. Therefore, the only way for $\varphi(i)$ to be empty is that none of the parents reshare the post, that is, $\xi(i)=\phi(i)$. Assuming independence, this occurs with probability
$\prod_{a \in \phi(i)} (1 - \theta_a)$,
which yields the desired expression. \\
\textsc{Layers $l >2$.} 
For deeper layers, the parents $j \in \phi(i)$ are not necessarily in the cascade. Thus, $i$ fails to join the cascade if every parent either does not belong to the cascade (nodes in $\psi(i)$) or fails to reshare (nodes in $\xi(i)$). 
To compute this probability, we partition the parent set according to the subset 
$\mathcal{B}\subseteq \phi(i)$ of parents that do not belong to the cascade. For a fixed subset $\mathcal{B}$, each parent $b\in \mathcal{B}$ is absent from the cascade with probability $1-\delta_b$, while each remaining parent $d\in \phi(i)\setminus \mathcal{B}$ belongs to the cascade but does not reshare, which occurs with probability $\delta_d(1-\theta_d)$. Under the independence assumption, the probability of this configuration is therefore
$$\prod_{b \in \mathcal{B}} (1 - \delta_b) \prod_{d \in \phi(i)\setminus \mathcal{B}} \delta_d (1 - \theta_d).$$
Summing over all such subsets yields the probability that no parent activates $i$  and subtracting from 1 gives the expression for $\delta_i$.
\end{proof}

\begin{figure}
\centering
    \includegraphics[width=0.35\textwidth]{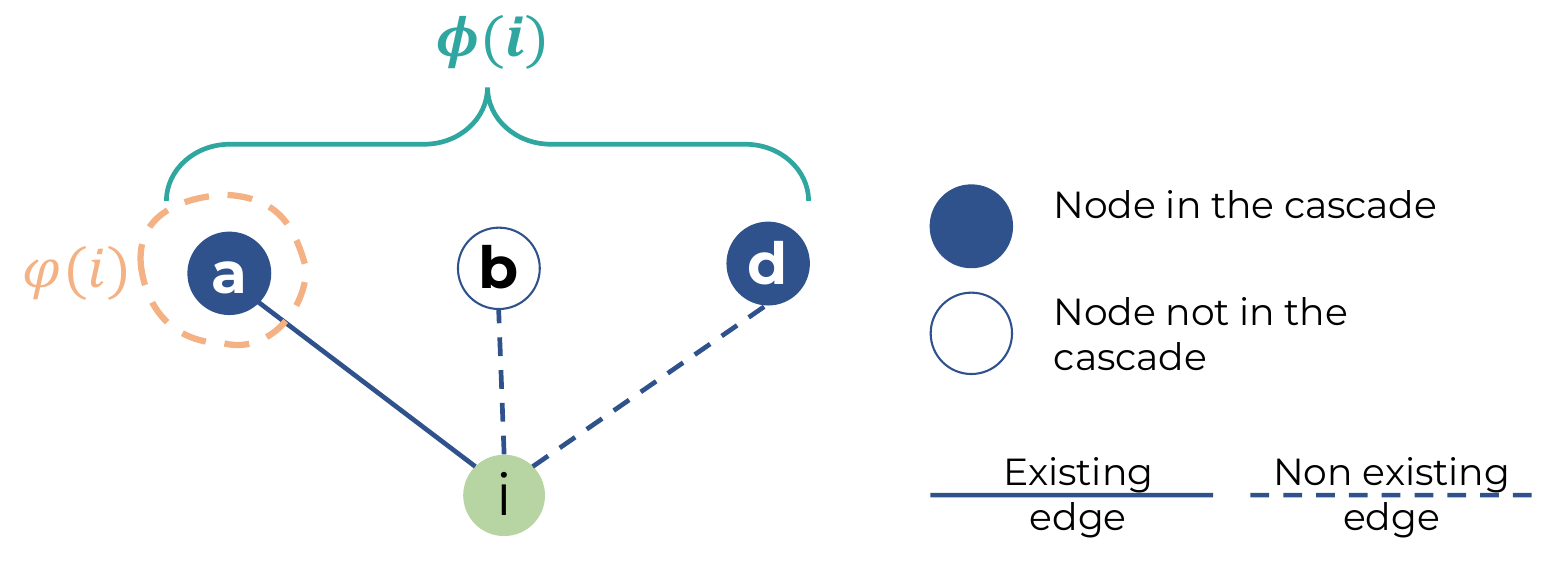}
    \vspace{-5pt}
    \caption{A node $i$ and its predecessors: $a\in\varphi(i)$ (resharing), $b\in\psi(i)$ (outside cascade), $d\in\xi(i)$ (in cascade but inactive). Only $a$ contributes to propagation.}
    \vspace{-10pt}
        \label{fig:predecessors}
\end{figure}

To compute the expected final opinion in the FJC model (which we do in Section~\ref{sec:fjc_solution}), we must first quantify the probability of \textit{which} specific subset of predecessors $\mathcal{A} \subseteq \phi(i)$ triggers an update on node $i$. The following lemma derives this probability $P(\left\{\varphi(i)=\mathcal{A}\right\})$.

\begin{lemma}[Probability of a specific predecessor set]\label{lemma:prob_a}
    Consider a node $ i \in \mathcal{O}_l^{(r_n)}$ and let $\mathcal{A}$ be a subset of $\phi(i)$ (i.e., $\mathcal{A} \in \mathcal{P}\left(\phi(i)\right)$). %Define the event $\mathsf{A}$  as ``nodes in $\mathcal{A}$ are those and only those that influence $i$'', which is given by:
    %$$\mathsf{A}=\left\{\varphi(i)=\mathcal{A}\right\}.$$
    The probability of event  $\left\{\varphi(i)=\mathcal{A}\right\}$  is given by:
    \begin{align} \vspace{-10pt}
        \beta(i,\mathcal{A}) &= \prod_{a\in\mathcal{A}}\delta_a \theta_a \cdot\nonumber\\
        &\cdot\sum_{\mathcal{B}\in\mathcal{P}(\phi(i)\setminus\mathcal{A})}\bigg[\prod_{b\in\mathcal{B}}(1-\delta_b)\cdot \prod_{d\in\phi(i)\setminus(\mathcal{A}\cup\mathcal{B})}\delta_d (1-\theta_d)\bigg]. 
        \label{eq:beta} 
    \end{align}
Note that it holds the following:
\begin{equation}
    \sum_{\mathcal{A}\in\mathcal{P}(\phi(i))\setminus\{\emptyset\}}  \beta(i,\mathcal{A}) = \delta_{i}, \vspace{-5pt}
    \label{eq:beta_c}
\end{equation}
given that all possible subsets of $\phi(i)$ but the empty set may lead to node $i$ becoming part of the cascade.
\end{lemma}

%The recursive structure of $\delta_i$ reflects the fact that cascade participation depends on the activation and resharing events of nodes in the previous layer. 
%Beyond determining whether a node participates in the cascade, we also need the probability that a specific subset of parents is responsible for activating it.

\begin{proof}
Fix a node $i$ and a subset $\mathcal{A} \subseteq \phi(i)$. We want to compute the probability that the set of active predecessors of $i$ is exactly $\mathcal{A}$, namely $\varphi(i)=\mathcal{A}$. \\
By definition, this event occurs if and only if every node $a \in \mathcal{A}$ belongs to the cascade and reshares the post, while every node in $\phi(i)\setminus \mathcal{A}$ does not belong to $\varphi(i)$. The latter can happen in two ways: either the node does not belong to the cascade, or it belongs to the cascade but does not reshare. 
To account for these possibilities, let $\mathcal{B} \subseteq \phi(i)\setminus \mathcal{A}$ denote the subset of parents that do not belong to the cascade. Then the remaining nodes in $\phi(i)\setminus (\mathcal{A}\cup \mathcal{B})$ are exactly those that belong to the cascade but do not reshare. 
Once $\mathcal{B}$ is fixed, the event $\{\varphi(i)=\mathcal{A}\}$ corresponds to the joint event
$\{\varphi(i)=\mathcal{A},\ \psi(i)=\mathcal{B},\ \xi(i)=\phi(i)\setminus (\mathcal{A}\cup \mathcal{B})\}.$
Assuming independence across nodes, its probability factorizes as
\[
\prod_{a\in \mathcal{A}}\delta_a\theta_a
\prod_{b\in \mathcal{B}}(1-\delta_b)
\prod_{d\in \phi(i)\setminus (\mathcal{A}\cup \mathcal{B})}\delta_d(1-\theta_d),
\]
where $\delta_a\theta_a$ is the probability that $a$ belongs to the cascade and reshares, $1-\delta_b$ is the probability that $b$ does not belong to the cascade, and $\delta_d(1-\theta_d)$ is the probability that $d$ belongs to the cascade but does not reshare.\\
Summing over all subsets $\mathcal{B} \subseteq \phi(i)\setminus \mathcal{A}$ yields Eq.~\eqref{eq:beta}.
Finally, summing $\beta(i,\mathcal{A})$ over all nonempty subsets $\mathcal{A}\subseteq \phi(i)$ gives $\delta_i$ (Eq.~\eqref{eq:beta_c}), since node $i$ belongs to the cascade if and only if $\varphi(i)\neq \emptyset$.
\end{proof}

\vspace{-20pt}
\subsection{Expected Opinion Dynamics}\label{sec:fjc_solution}

Having established a matrix formula for how opinions update when a specific cascade realization occurs (Section~\ref{sec:matrix_rep}, Corollary~\ref{cor:end_opinion_layer}) and a probabilistic description of which cascade realizations are likely (Section~\ref{sec:probabilistic_structure}, Proposition~\ref{pro:deltai} and Lemma~\ref{lemma:prob_a}), we now combine these to compute the expected opinion trajectory, i.e.,  the average evolution of opinions over all possible cascade realizations.
Note that the stochastic nature of the FJC model prevents us from obtaining a deterministic final opinion vector, unlike the FJ model. 
% To directly apply the results of Section~\ref{sec:modelling}, one would need the exact propagation trees $\mathcal{T}^{(r_n,\theta)}$, which are only known after the cascade has fully unfolded. 
% Instead, we focus on estimating the \emph{expected} final opinion, leveraging the probabilistic framework developed in Section~\ref{sec:probabilistic_structure}.

\subsubsection{Solution of a single cascade}
\label{sec:single_cascade}

We begin by analyzing the propagation of a single post, $p_n$, originated at node $r_n$, to estimate the expected final opinion at the end of the corresponding cascade. Let $\bm{y}^{(n)}_0$ denote the opinion vector at the beginning of the cascade, and $\bm{y}^{(n)}_l$ the opinion vector at the end of layer $l$. 

A tempting simplification would be to express the expected opinion at layer $l$, namely $E[\bm{y}_l^{(n)}]$, as a function of the expected opinions in the previous layer, i.e., $E[\bm{y}_{l-1}^{(n)}]$. %$E[\bm{y}_0^{(n)}],\dots,E[\bm{y}_{l-1}^{(n)}]$. 
However, such a closed recursion cannot be obtained exactly, due to the statistical coupling between the diffusion process and the opinion dynamics. %Consider the example in Figure~\ref{fig:example_tree}. 
Assume that we want to express the expected opinion of a node $i$, belonging to layer $l$, as a function of the expected opinion of node $j$, which belongs to layer $l-1$. Assume that the update occurs at time $t$. By the law of Total Probability, we should write:
\begin{align} \vspace{-10pt}
    E[x_i(t)]&=\sum_{\mathcal{A}\in\mathcal{P}(\{j\})}E[x_i(t)|\varphi(i)=\mathcal{A}] \, P\{\varphi(i)=\mathcal{A}\}\\
    &=E[x_i(t-1)] \, P\{\varphi(i)=\emptyset\} + \\
    &\qquad +E[x_i(t)|\varphi(i)=\{j\}] \, P\{\varphi(i)=\{j\}\},\label{eq:example}
\end{align}
where in the second equality we used that in case $\varphi(i)=\emptyset$, the opinion of node $i$ remains the same. 
A natural way to compute the second addend is to evaluate the opinion propagation of node $j$ to node $i$, using Eq.~\eqref{eq:fjc_update} after passing by expectation. However, this is not valid: the expected opinion of node $j$, $E[x_j(t-1)]$, averages over realizations in which node $j$ may not belong to the cascade, which is inconsistent with the conditioning event $\varphi(i)=\{j\}$ that holds in the second addend of Eq.~\eqref{eq:example}. This shows that the expectation in the second term cannot be expressed in terms of the unconditional expectation $E[x_j(t-1)]$. 
Hence, expected opinions across consecutive layers cannot be decoupled, and a closed recursion in unconditional expectations is not available. 
The correct formulation must therefore explicitly account for the cascade history through conditional expectations. This is formally proved in the following lemma, whose proof is in the SI.

\begin{lemma}[Exact conditional layer-wise recursion]
\label{lemma:exact_conditional_layer_recursion}
Let $\mathcal{H}_{l-1}$ denote the $\sigma$-algebra generated by all activation and resharing events in layers $1,\dots,l-1$. 
Then the expected opinion vector at the end of layer $l$ satisfies
\begin{equation*}
E[\bm y_l] \;=\; E\!\left[\,E[\bm y_l\mid \mathcal{H}_{l-1}]\,\right].
\label{eq:iterated_expectation_layer}
\end{equation*}
Moreover, conditional on $\mathcal{H}_{l-1}$ one has the exact identity
\begin{align}
E[\bm y_l\mid \mathcal{H}_{l-1}]
&=
\bm y_{l-1}
+
\sum_{i\in \mathcal{O}_l}
\sum_{\mathcal{A}\in \mathcal{P}( \phi(i))\setminus\emptyset}
\beta_{l}(i,\mathcal{A}\mid \mathcal{H}_{l-1})\cdot\;\\
&\cdot
\Bigg(
\Big(\tfrac{1}{|\mathcal{A}|}\sum_{j\in \mathcal{A}}A_{ij}+B_i-I\Big)\bm y_{l-1}
+
C_i\bm u
\Bigg)
\label{eq:exact_cond_layer}
\end{align}
where $\phi(i)$ is the (deterministic) parent set of $i$, $\varphi(i)\subseteq \phi(i)$ is the (random) set of resharing parents of $i$, and
\[
\beta_{l}(i,\mathcal{A}\mid \mathcal{H}_{l-1})
:=P(\varphi(i)=\mathcal{A}\mid \mathcal{H}_{l-1}).
\]
\end{lemma}

Lemma~\ref{lemma:exact_conditional_layer_recursion} shows that %the evolution of the expected opinion at layer $l$ can be written exactly in terms of the conditional expectation $E[\bm y_l \mid \mathcal H_{l-1}]$. 
%However, this representation is not closed in terms of unconditional expectations, since the conditional law of $\bm y_{l-1}$ depends on the cascade history through the random sets  $\varphi(i)$. 
%Consequently, 
an exact recursion involving only $E[\bm y_{l-1}]$ is not available. 
In order to obtain a tractable inter-layer recursion, we introduce a mean-field (decoupling) approximation (MFA) that neglects the statistical dependence between the resharing configuration $\varphi(i)$ and the opinion vector $\bm y_{l-1}$: 
\begin{equation}
\textbf{(MFA): }  E\!\left[\bm{y}_{l-1}\mid \varphi(i)=\mathcal{A}\right]
\approx 
E[\bm{y}_{l-1}],
\label{eq:mf_decoupling}
\end{equation}
for all $i\in \mathcal{O}_l$ and for all $\mathcal{A}\in\mathcal{P}(\phi(i))\setminus\emptyset$. 

% The MFA replaces the conditional expectation of the previous-layer opinion vector by its unconditional counterpart, effectively neglecting the correlation between cascade membership and opinion evolution.  
% More precisely, without approximation the unconditional expectation of a component $y_{l-1,a}$  for a node $a$, can be decomposed as
% \[
% E[y_{l-1,a}]
% =
% \delta_a\,E[y_{l-1,a}\mid a\in\mathcal C]
% +
% (1-\delta_a)\,y_{0,a},
% \]
% where $\delta_a$ is the probability that node $a$ belongs to the cascade. 
% Conditioning on $\{\varphi(i)=\mathcal A\}$ removes the second term for all $a\in\mathcal A$, since $a\in\mathcal C$ is enforced. 
% The MFA assumes that this difference is negligible at the level of expectations, thereby closing the recursion in terms of $E[y_{l-1}]$ alone. 
% From a modeling perspective, this corresponds to treating diffusion and opinion dynamics as weakly correlated at the expectation level, while still preserving their structural interaction in the update rule.

Under this approximation, the conditional recursion of Lemma~\ref{lemma:exact_conditional_layer_recursion} reduces to a closed linear recursion depending only on the marginal resharing probabilities and on $E[\bm y_{l-1}]$, as stated in the next theorem.

\begin{theorem}[Approximate closed recursion for the expected opinions at the end of each layer]
\label{theo:approx_expected_layer_recursion}

Let $E[\bm{y}^{(n)}_0]$ denote the expected opinion vector when post $p_n$ is published, i.e. the opinion vector at the beginning of the $n$-th cascade, obtained as the final state of the previous cascade (or equal to $\bm{u}$ for $n=1$). For each layer $l\geq 1$, assume the mean-field decoupling approximation (MFA) as stated in Eq.~\eqref{eq:mf_decoupling}.
Then, the expected opinion at the end of layer $l$ satisfies the approximate linear recursion
\begin{align}
E[\bm{y}_{l}^{(n)}] 
&= 
\tilde{A}_l^{(r,\theta)}\,E[\bm{y}_{l-1}^{(n)}]
+\tilde{C}_l^{(r,\theta)}\,\bm{u}
+\tilde{D}_l^{(r,\theta)}\,E[\bm{y}_0],
\label{eq:approx_layer_recursion}
\end{align}
where
\begin{align}
\tilde{A}_{l}^{(r,\theta)}
&=
\sum\limits_{i\in \mathcal{O}_l^{(r_n)}}
\sum\limits_{\mathcal{A}\in\mathcal{P}(\phi(i))\setminus\{\emptyset\}}
\frac{\beta(i,\mathcal{A})}{|\mathcal{A}|}\sum_{a\in\mathcal{A}}A_{i a}
+\\
&\qquad\qquad+
\left(I-\sum\limits_{i\in\mathcal{O}_l^{(r_n)}} e_{i}e_{i}^T\right),
\label{eq:a_tilde_new}\\
\tilde{C}_{l}^{(r,\theta)}
&=
\sum\limits_{i\in\mathcal{O}_l^{(r_n)}}\delta_{i}C_{i},
\label{eq:c_tilde_new}\\
\tilde{D}_{l}^{(r,\theta)}
&=
\sum\limits_{i\in\mathcal{O}_l^{(r_n)}}(1-\delta_{i}) e_{i}e_{i}^T .
\label{eq:d_tilde_new}
\end{align}
\end{theorem}

\begin{proof}
For notational simplicity, we omit the superscripts related to $r_n$, $\theta$, and $n$, which are considered fixed throughout the proof.
From Corollary~\ref{cor:end_opinion_layer}, for a fixed cascade realization it holds that
\begin{equation}
\bm{y}_l
=
\sum_{i\in \mathcal{T}_l}
\frac{1}{|\varphi(i)|}\sum_{j\in \varphi(i)}A_{ij}\,\bm{y}_{l-1}
+
\tilde{B}_{\mathcal{T}_l}\,\bm{y}_{l-1}
+
\sum_{i\in \mathcal{T}_l} C_i\,\bm{u},
\label{eq:realized_layer_update}
\end{equation}
where $\mathcal{T}_l$ is the (random) set of nodes in the cascade at layer $l$ and $
\tilde{B}_{\mathcal{T}_l}=I-\sum_{k\in \mathcal{T}_l} e_ke_k^\top$.
Taking expectations in Eq.~\eqref{eq:realized_layer_update} and using linearity yields
\begin{align}
E[\bm{y}_l]
&=
E\!\left[
\sum_{i\in \mathcal{T}_l}
\frac{1}{|\varphi(i)|}\sum_{j\in \varphi(i)}A_{ij}\,\bm{y}_{l-1}
\right]
+
E[\tilde{B}_{\mathcal{T}_l}\,\bm{y}_{l-1}]
+\\
&\qquad+
E\!\left[\sum_{i\in \mathcal{T}_l} C_i\,\bm{u}\right].
\label{eq:split_terms}
\end{align}
Now we will consider the terms in Eq.~\eqref{eq:split_terms} separately and prove the claim. \\
Consider the first that accounts for $E[\bm{y}_{l-1}]$. For each fixed $i\in \mathcal{O}_l$, the law of total expectation over the random set $\varphi(i)$ gives the identity
\begin{align}
E\!\Big[&
\frac{1}{|\varphi(i)|}\sum_{j\in \varphi(i)}A_{ij}\,\bm{y}_{l-1}
\Big]
=\\
&=
\sum_{\mathcal{A}\in\mathcal{P}(\phi(i))\setminus\{\emptyset\}}
\beta(i,\mathcal{A})\;
\frac{1}{|\mathcal{A}|}\sum_{j\in \mathcal{A}}
A_{ij}\,
E\!\left[\bm{y}_{l-1}\mid \varphi(i)=\mathcal{A}\right].\nonumber
\label{eq:total_over_varphi}
\end{align}
Applying the MFA (Eq.~\eqref{eq:mf_decoupling}) we obtain
\begin{align}
E\!\Big[&
\frac{1}{|\varphi(i)|}\sum_{j\in \varphi(i)}A_{ij}\,\bm{y}_{l-1}
\Big]
\approx\\
&\approx
\sum_{\mathcal{A}\in\mathcal{P}(\phi(i))\setminus\{\emptyset\}}
\beta(i,\mathcal{A})\;
\left(\frac{1}{|\mathcal{A}|}\sum_{j\in \mathcal{A}}A_{ij}\right)
E[\bm{y}_{l-1}].
\label{eq:influence_decoupled}
\end{align}
Summing Eq.~\eqref{eq:influence_decoupled} over $i\in \mathcal{O}_l$ provides the first term in Eq.~\eqref{eq:a_tilde_new}.

Next, consider the second term in Eq.~\eqref{eq:split_terms}. Since $\tilde{B}_{\mathcal{T}_l}=I-\sum_{k\in \mathcal{T}_l} e_ke_k^\top$, it follows that
\begin{align*}
\tilde{B}_{\mathcal{T}_l}\bm{y}_{l-1}
&=
\bm{y}_{l-1}
-
\sum_{k\in \mathcal{T}_l} e_ke_k^\top \bm{y}_{l-1}\nonumber\\
&=
\bm{y}_{l-1}
-
\sum_{k\in \mathcal{O}_l}  \mathbf 1_{\{k\in \mathcal{T}_l\}}\,e_ke_k^\top \bm{y}_{l-1},\nonumber
\end{align*}
where $\mathbf 1_{\{k\in \mathcal{T}_l\}}$ denotes the scalar indicator random variable taking value $1$ if $k\in \mathcal{T}_l$ and $0$ otherwise.
Taking expectations and using linearity gives
\begin{equation}
E\left[\tilde{B}_{\mathcal{T}_l}\bm{y}_{l-1}\right]
=
E[\bm{y}_{l-1}]
-
\sum_{k\in \mathcal{O}_l}  E\left[\mathbf 1_{\{k\in \mathcal{T}_l\}}\,e_ke_k^\top \bm{y}_{l-1}\right].
\label{eq:B_tilde_0}
\end{equation}
In order to separate the components corresponding to $\mathcal{T}_l$ from those corresponding to $\mathcal{O}_l\setminus\mathcal{T}_l$ in the second term of Eq.~\eqref{eq:B_tilde_0}, we add and subtract $\sum_{k\in \mathcal{O}_l} e_k e_k^\top E[\bm{y}_{l-1}]$, obtaining
\begin{align}
E\big[\tilde{B}_{\mathcal{T}_l}\bm{y}_{l-1}\big]
&=
\left(I-\sum_{k\in \mathcal{O}_l}e_ke_k^\top\right)E[\bm{y}_{l-1}]
+ \\
&\quad+\sum_{k\in \mathcal{O}_l}\!\Big(
e_ke_k^\top E[\bm{y}_{l-1}]
\!-\!E\!\left[\mathbf 1_{\{k\in \mathcal{T}_l\}}\,e_ke_k^\top \bm{y}_{l-1}\right]
\!\Big)
\nonumber\\
&=
\left(I-\sum_{k\in O_l}e_ke_k^\top\right)E[\bm{y}_{l-1}]\\
&\quad
+
\sum_{k\in O_l}e_ke_k^\top\Big(
E[\bm{y}_{l-1}]
-E\!\left[\mathbf 1_{\{k\in \mathcal{T}_l\}}\,\bm{y}_{l-1}\right]
\Big).
\label{eq:B_addsub}
\end{align}
Moreover, it is straightforward to observe that:
$$
E\!\left[\mathbf 1_{\{k\in \mathcal{T}_l\}}\,\bm{y}_{l-1}\right]
=
P\{k\in \mathcal{T}_l\}\,E\!\left[\bm{y}_{l-1}| k\in \mathcal{T}_l\right].
$$
With this trick, we can effectively separate the cases in which a node $k\in O_l$ is in the cascade or not and  exploit the layer-based update semantics. In fact, if $k\in O_l$ but $k\notin \mathcal{T}_l$ (i.e., $k$ does not belong to the cascade), then $k$ does not update and its component remains equal to its initial value, namely
\begin{equation}
e_ke_k^\top \bm{y}_{l-1}=e_ke_k^\top \bm{y}_0
\qquad\text{on the event } \{k\notin \mathcal{C}\}.
\label{eq:inactive_keeps_initial}
\end{equation}
Since $k\in \mathcal{O}_l$ is a node in layer $l$, the event $\{k\in \mathcal{T}_l\}$ is equivalent to $\{k\in \mathcal{C}\}$ and therefore
$$
\delta_k=P\{k\in \mathcal{C}\}=P\{k\in \mathcal{T}_l\}.
$$
Using the partition $\{k\in \mathcal{C}\}\cup\{k\notin \mathcal{C}\}$ together with Eq.~\eqref{eq:inactive_keeps_initial}, we obtain that:
\begin{align}
e_ke_k^\top E[\bm{y}_{l-1}]
&=
e_ke_k^\top\Big(\delta_k E[\bm{y}_{l-1}| k\in \mathcal{C}]+\\
&\qquad\qquad+(1-\delta_k) E[\bm{y}_{l-1}| k\notin \mathcal{C}]\Big)
\nonumber\\
&=
\delta_k\, e_ke_k^\top E[\bm{y}_{l-1}| k\in \mathcal{T}_l]
+(1-\delta_k)\,e_ke_k^\top E[\bm{y}_0].
\label{eq:kth_coord_partition}
\end{align}
Substituting this term into Eq.~\eqref{eq:B_addsub}, the terms involving
$\delta_k\,e_ke_k^\top E[\bm{y}_{l-1}| k\in \mathcal{T}_l]$
cancel with
$e_ke_k^\top E[\mathbf 1_{\{k\in \mathcal{T}_l\}}\bm{y}_{l-1}]$,
and we obtain
\begin{align}
E[\tilde{B}_{\mathcal{T}_l}\bm{y}_{l-1}]
&=
\left(I-\sum_{k\in O_l}e_ke_k^\top\right)E[\bm{y}_{l-1}]
+\\
&\qquad\qquad +
\sum_{k\in O_l}(1-\delta_k)\,e_ke_k^\top\,E[\bm{y}_0].
\end{align}
The first term above corresponds to the second addend in Eq.~\eqref{eq:a_tilde_new}, while the second term yields the contribution $\tilde{D}_l^{(r,\theta)}E[\bm{y}_0]$ in Eq.~\eqref{eq:approx_layer_recursion}.

Finally, consider the last term in Eq.~\eqref{eq:split_terms}. Since $C_i$ is deterministic,
$E\!\left[\sum_{i\in \mathcal{T}_l} C_i\,\bm{u}\right] = \sum_{i\in O_l}\delta_i\,C_i\,\bm{u}$,
which provides the term $\tilde{C}_l^{(r,\theta)}\bm{u}$ in Eq.~\eqref{eq:approx_layer_recursion}.

After considering the three contributions, we ultimately derived Eq.~\ref{eq:approx_layer_recursion}, which concludes the proof.
\end{proof}

% \vspace{-7pt}
\subsubsection{Solution of a sequence of cascades}
\label{sec:sequence_cascade}

Theorem~\ref{theo:approx_expected_layer_recursion} directly leads to the following corollary, which describes how to solve the FJC model when a total amount of $N$ posts are generated and spread across the network.

\begin{theorem}
    [Solution of the FJC model]\label{cor:fjc_series_cascades}
    Under the mean-field decoupling approximation (MFA) introduced in Theorem~\ref{theo:approx_expected_layer_recursion}, when a total of $N$ posts ($p_1, \dots, p_N$) are generated, the overall expected final opinion $E[\bm{z}^{(N)}]$ is computed by iteratively applying Theorem~\ref{theo:approx_expected_layer_recursion} for $n=1, \dots, N$. The output of one cascade serves as the initial condition for the next: $$E[\bm{y}_0^{(n)}] = \begin{cases} \bm{u} & n=1 \\ E[\bm{y}_L^{(n-1)}] & n=2,\dots,N. \end{cases}$$    
    % Let us consider the FJC model applied to the graph $\mathcal{G}$ where $N$ posts $p_1,\dots,p_{N}$ are generated by the nodes $r_1,\dots,r_{N}$. Let us denote with $\mathcal{C}^{(r_n)}$ the cascade corresponding to the post $p_n$. If $\bm{u}$ is the initial opinion vector, %initial  opinion $E[\bm{y}_{0}^{(r_i,\theta)}]$ of nodes is $\bm{u}$ if $i=1$ and is the expected final opinion of the previous cascade, i.e.:
    % %\begin{equation}
    % %    E[\bm{y}_0^{(r_i,\theta)}]= \left\{
    % %    \begin{array}{ll}
    % %        \bm{u} &  \text{ if $i=1$}\\
    % %        \\
    % %         E[\bm{y}_{T}^{(r_{i-1},\theta)}]\ & \text{otherwise}
    % %    \end{array}
    % %    \right.
    % %\end{equation}
    %  the expected final opinion $E[\bm{z}^{N}]=E[\bm{y}^{N}_L]$ at the end of the process is computed applying iteratively Theorem~\ref{theo:fjc_single_cascade} for $n=1,\dots,N$   with the following initial setting:
    % \begin{equation}
    % E[\bm{y}_0^{(n)}]= \begin{cases}
    %        \bm{u} & n=1\\
    %        E[\bm{y}_L ^{(n-1)} ] & n=2,\dots,N-1\\
    %     \end{cases}
    % \end{equation}
    Thus, the overall expected final opinion after $N$ cascades is obtained by sequentially propagating the approximate layer recursion.
\end{theorem}

\vspace{-10pt}
\subsection{Scalable Approximations}
\label{sec:approximations}

The theoretical framework developed in this work provides an accurate characterization of diffusion and opinion propagation in complex networks. Its precision, however, comes at a non-negligible computational cost, as it requires evaluating all possible diffusion paths, each corresponding to a potential realization of the stochastic process. %While the combinatorial nature increases computational effort, it remains feasible for networks of up to a few hundred nodes. 
This computation is performed once for fixed parameters $r$, $\Lambda$, and $\theta$, allowing the resulting matrices to be stored and reused for rapid evaluation of final opinions from initial conditions, making the framework suitable for extensive  \emph{what-if} analyses.

To enhance scalability to larger systems, we replace the full cascade structure with a tree-based approximation, illustrated in Figure~\ref{fig:example_approx} on the cascade structure of Figure~\ref{fig:example}.(b). In the original network, a node may be reachable from the cascade root through multiple parents and multiple diffusion paths, which induces the combinatorial enumeration of predecessor sets. The approximation retains, for each reachable node, only the predecessor lying on the most probable path from 
the root. As a result, each node has at most one parent in the approximated  cascade structure, so its predecessor set is either empty or a singleton.
%To enhance scalability to larger systems, we introduce approximation strategies that preserve the key probabilistic and structural properties of the diffusion process while reducing its combinatorial complexity. The core idea is to approximate the FJC diffusion dynamics with a single, most probable propagation tree, in which each node is connected to the root through a unique path corresponding to the most likely influence trajectory in the original network. Consequently, for any node $i$ its predecessor set $\varphi(i)$ is either empty or contains a single element. %The symbol $\hat{}$ indicates that here we are considering the propagation process over an approximated graph and note on the full tree $\mathcal{T}^{(r,\theta)}$.

We now describe the construction of this most probable influence tree, for the graph $G$, rooted at a specific source node $r$. Our main idea is to adapt Dijkstra’s algorithm to compute the shortest weighted paths from the root to all other reachable nodes. To this aim, we assigned probability $1$ to all edges $(r,i)$, which represent the propagation from the root to its neighbors, consistent with the FJC model. For each edge $(i,j)\in \mathcal{E}$ of the graph, we associate a propagation probability $\theta_i$, which reflects the cascade propagation process in the FJC model from node $i$ to $j$. Therefore, the edge weights are thus defined as $-\log(\theta_i)$.  
For each node $i$, we denote by $d_i$  the distance from $r$ to $i$ and with $\sigma(r,i)=(v_0,\dots,v_{d_i})$ the unique shortest path from $r$ to $i$, where $v_0=r$ and $v_{d_i}=i$. The reachability probability $\rho_i$ of $i$ is given by $\rho_i=\prod_{k=1}^{d_{i}-1}\theta_k$. %, where $d_i$ is the shortest-path distance from $r$ to $i$. 
%Immediate neighbors of the root are included by design with a reachability probability of 1, consistent with the FJC model (Remark~\ref{rem:layers12}). 

For all  nodes, inclusion in the tree depends on one of two variants of the method described below:
 
\begin{description}
    \item[Full probability tree (FT):] includes all nodes reachable from $r$, irrespective of their reachability probability;%the full probability tree, which includes all reachable nodes regardless of reachability, 
    \item[Thresholded tree (TT):] retains only nodes whose reachability probability exceeds a user-defined threshold $\tau$, that is, nodes satisfying $\rho_i=\prod_{k=1}^{d_{i}-1}\theta_k \geq \tau$.%, where $d_i$ is their shortest-path distance from $r$.%which retains only nodes whose reachability exceeds a user-defined threshold $\tau$ ($\theta^{d_i-1} \geq \tau$).
\end{description}
It is important to note that in both cases, the immediate neighbors of the root are included by design, as their reachability probability $\rho_i =1$ for all neighbors $i$. This is consistent with the FJC model. In both cases, the resulting structure is a directed acyclic tree that encodes the most probable parent-child relations between nodes, together with their reachability probabilities derived from their shortest-path distances.

The approximated trees, FT and TT, lead to lower computational costs when calculating results. This reduction primarily arises from the fact that a node $i$ has at most one possible predecessor ($|\overset{FT}{\varphi}(i)|,|\overset{TT}{\varphi}(i)|\leq 1$), which simplifies the combinatorial calculations. The main results presented in Section~\ref{sec:fjc_solution} easily carry over to the approximated trees setting. For lack of space, please refer to SI for the details.
%In the following, we highlight how the main results presented in Section~\ref{sec:fjc_solution}, map into the approximated trees setting.

The complexity of the methods are summarized in Table~\ref{tab:complexity_comparison}. In both variants, we can see that the proposed approximation represents a drastic improvement over the exponential number of combinations required by the exact method, enabling the analysis of networks orders of magnitude larger. 
For a detailed derivation of these values, please refer to SI. 
%Constructing the most probable influence tree requires:
%(i) computing shortest paths from $r$ to all other nodes via Dijkstra’s algorithm with non-negative weights, whose time complexity is $O(E \log V)$ when using a binary heap, and
%(ii) extracting the tree structure, which is linear in the number of edges retained.
%Thus, the overall cost scales as $O(E \log V)$ for the Full Tree (FT) construction. For the Thresholded Tree (TT), the complexity reduces as a function of how aggressively the threshold $\tau$ prunes low-probability branches. In sparse networks, where each node has a bounded number of neighbors, and when the threshold $\tau$ prunes most low-probability branches, the number of edges explored scales linearly with the number of surviving nodes ($V_\tau$), making the runtime approximately $O(V_\tau)$ in practice. 
%In both variants, this represents a drastic improvement over the exponential number of combinations required by the exact method, enabling the analysis of networks orders of magnitude larger.

\begin{table*}[t!]
\centering
\caption{Comparison of computational strategies}
\label{tab:complexity_comparison}
\begin{tabular}{lccc}
\toprule
\textbf{Method} & \textbf{Complexity} & \textbf{Scalability} & \textbf{Accuracy}\\

MFA-based solution & Exponential in $V$ & Limited to small/medium graphs & High (mean-field approx.) \\
Full Tree (FT) & $O(E \log V)$ & Scales to large graphs & High (minor structural approx.) \\
Thresholded Tree (TT) & Up to $O(E \log V)$
$\rightarrow$ $\approx O(V)$ with aggressive pruning  &
Scales to very large graphs & Tunable vs.~threshold $\tau$ \\
\bottomrule
\end{tabular}
\vspace{-10pt}
\end{table*}

\begin{figure}[t]
    \centering
      \subfloat[FT approximation]{
          \includegraphics[width=0.2\textwidth]{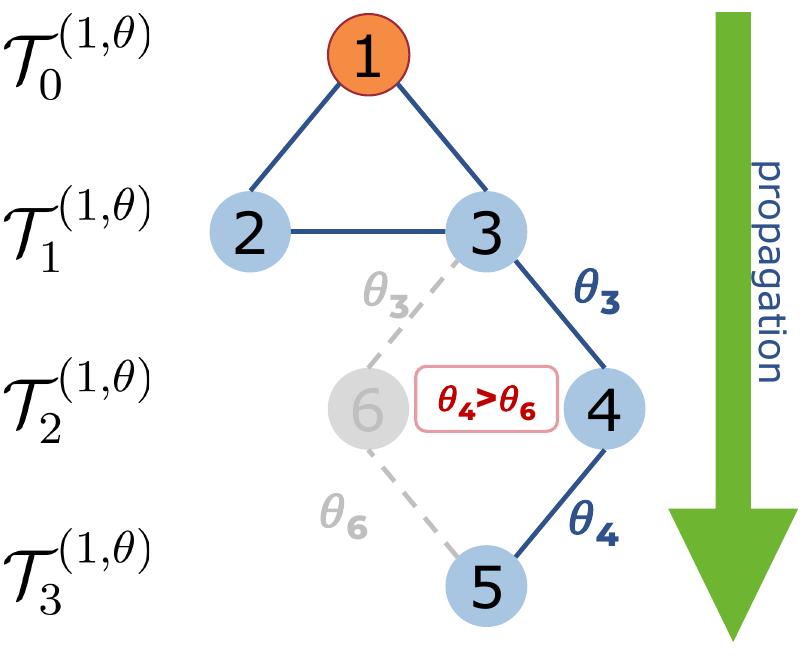}
      }\hfill
      \subfloat[TT approximation]{
          \includegraphics[width=0.2\textwidth]{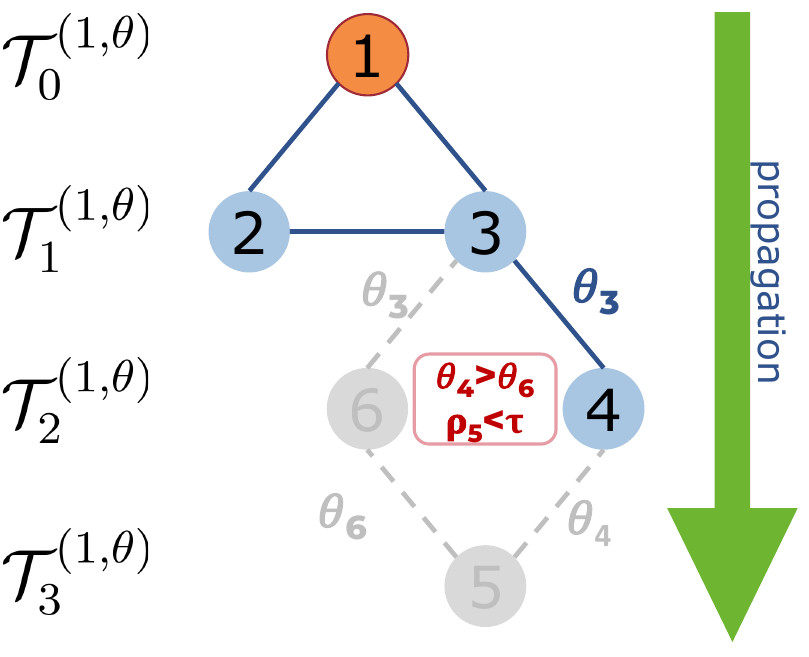}
      }
      \caption{Tree-based approximations applied to the cascade structure of Figure~\ref{fig:example}.(b).  (a) FT keeps the most probable predecessor for each node; since $\theta_4>\theta_6$, node~5 is attached to node~4. 
      (b) TT additionally prunes nodes below the reachability threshold; since $\rho_5=\theta_3\theta_4<\tau$, node~5 is removed.}
    \label{fig:example_approx}
    \vspace{-10pt}
\end{figure}

\vspace{-10pt}
\section{Performance Evaluation}
\label{sec:performance_evaluation}

We now evaluate the accuracy and implications of the theoretical framework. The analysis pursues two goals: validating the numerical and tree-based approximations, and examining how cascades of information alter the opinion dynamics and, consequently, influence polarization. 
% Compared to our previous work~\cite{biondi2026cascade}, here we focus on a more complex network structure, specifically a two-block SBM, which enables us to investigate how different communities can condition one another. 
We consider a two-block Stochastic Block Model (SBM) with 100 nodes, within- and between-community connection probabilities $p_{in}=0.3$ and $p_{out}=0.01$ (Figure~\ref{fig:graph}). The social matrix $\hat{W}$ coincides with the adjacency matrix of the network. As described in Section~\ref{sec:social_network}, the influence matrix is derived by row-normalizing $\hat{W}$: $w_{ij} = \frac{\hat{w}{ij}}{\sum_k \hat{w}{ik}}$. This topology offers a demanding benchmark: approximation errors accumulate over multiple cascade layers, and the strong modularity ($p_{out}\ll p_{in}$) yields longer paths than in Erdos–Renyi or scale-free networks (see SI for a detailed discussion). At the same time, the SBM structure provides a clear community organization that makes it easier to interpret a wide range of initial opinion configurations, including but not limited to homophilic ones, thereby enabling a controlled “what-if” analysis of how cross-community alignments shape the amplification or attenuation of influence through cascades.%At the same time, the SBM structure allows us to impose a wide range of initial opinion configurations—including but not limited to homophilic ones—thereby enabling a controlled “what-if’’ analysis of how different cross-community alignments shape the amplification or attenuation of influence through cascades. %At the same time, the SBM allows explicitly homophilic initial conditions and thereby enables a controlled “what-if’’ analysis of how cascades amplify or mitigate cross-community influence.

\begin{figure}[t!]
     \centering
     \includegraphics[width=0.45\textwidth]{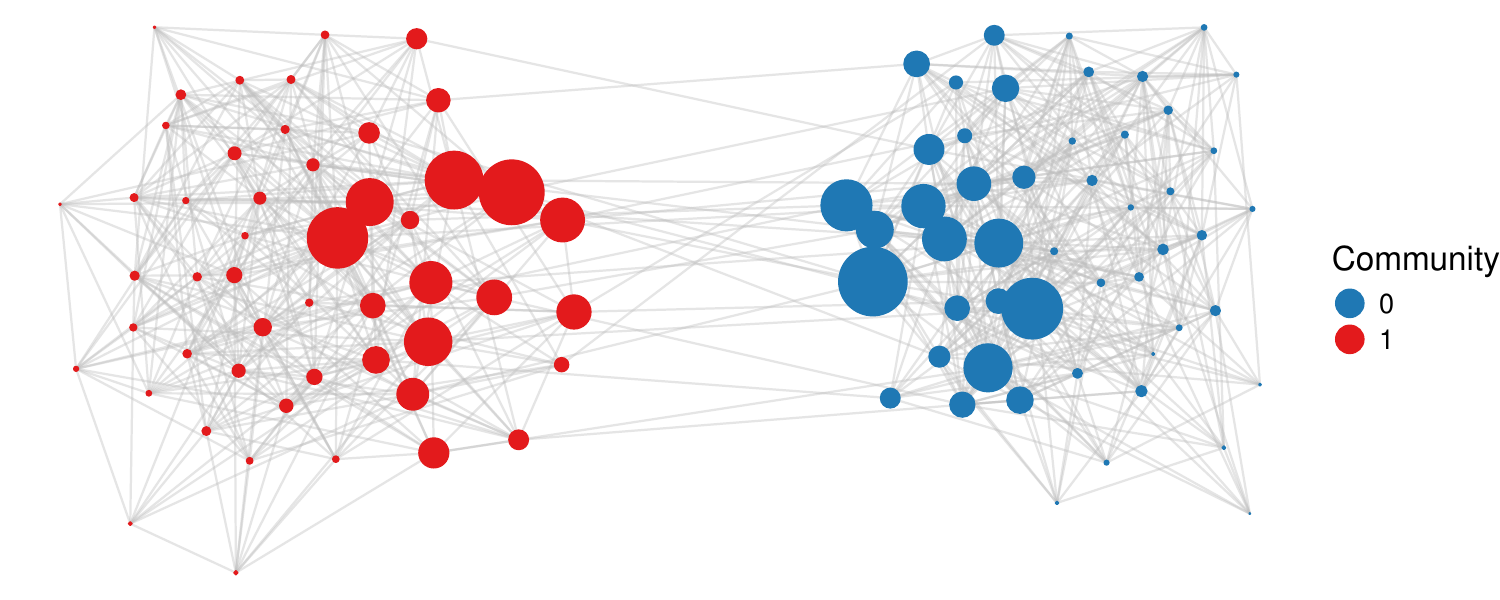}
    \caption{SBM used in the performance evaluation. Node size is proportional to the betweenness centrality.}
    \label{fig:graph}
    \vspace{-5pt}
\end{figure}

%

% \begin{table}
% \centering \vspace{-7pt}
% \caption{Statistics of the SBM network}\label{tab:centralities}
% \vspace{-5pt}
% \setlength{\tabcolsep}{2pt}
% \begin{tabular}{ ccccccc }
% \toprule
%  & \#nodes & \#bridge nodes &  \multicolumn{2}{c}{Betweenness}&  \multicolumn{2}{c}{Degree}\\
%  \cmidrule(l){4-5} 
%  \cmidrule(l){6-7} 
%  &&&Mean   & Std  &  Mean   & Std \\
%  \midrule
% \textbf{Comm. 0} & 50    & 22   & 0.0128  & 0.0120  & 15.2 & 3.48\\
% \textbf{Comm. 1} & 50  &  19  & 0.0130  & 0.0117  & 14.1 & 2.98\\
% \bottomrule
% \end{tabular} \vspace{-10pt}
% \end{table}

% \begin{figure}[t!]
%      \centering
%      \includegraphics[width=0.45\textwidth]{figures/centAll_SBM_2communities.pdf} \vspace{-7pt}
%     \caption{Betweenness and degree centrality in the different communities: (A) distribution pattern, (B) scatterplot visualization}
%     \vspace{-12pt}
%     \label{fig:centralities}
% \end{figure}

We consider resharing probabilities uniform across all nodes $\theta_i=\theta$ for all $i$, with $\theta\in{0.1,0.5,1}$ and we randomly select $15$ cascade sources. To study both coupled and decoupled susceptibility-centrality relationships, we set the susceptibility values for the nodes in three different ways: either directly or indirectly proportional to betweenness centrality $\mathrm{B}_i$, or equal to \(0.6\) for all nodes. For each setting, FJC final opinions refer to the expected opinions after the full sequence of selected cascade sources, with cascades applied sequentially as in Theorem~\ref{cor:fjc_series_cascades}.
%For each setting, we compute the theoretical outcomes of Theorem~\ref{cor:fjc_series_cascades} under MFA. %of Eq.~\eqref{eq:weaker_app}. 
The only demanding step is the computation of $\beta(i,\mathcal{A})$ for $i$ in deeper cascade layers; these values are precomputed and reused. Please note that we have omitted results for $\theta=1$ due to their qualitative similarity with the other cases. Interested readers are referred to the SI for the complete set of results.
 
Our primary objective is to analyze how the cascade process shapes opinion dynamics and homophily. Specifically, we assess whether initial opinion configurations evolve into more or less polarized final states when the opinion dynamics are coupled with a cascade process, compared to when they are not. Here, we interpret \emph{polarization} as the presence of distinct, well-separated communities holding opposing opinions.  To this end, we define four initial opinion settings. In \textsc{polarized communities} (\textsc{PC}), the two communities start with opposite opinions ($+1$ and $-1$), representing a fully polarized state. In \textsc{neutral bridges} (\textsc{NB}), community nodes again hold opposite opinions, but bridge nodes (connecting members of both communities) are assigned neutral values ($0$). In \textsc{polarized bridges} (\textsc{PB}), the situation is reversed: bridge nodes hold extreme opinions ($+1$ or $-1$) based on their community affiliation, while all other nodes remain neutral. Finally, in \textsc{asymmetric polarization} (\textsc{AP}), one community is fully neutral ($0$) and the other fully extreme ($+1$), representing an unbalanced opinion landscape.  

The interested reader can find additional information on the SBM network, along with a visualization of the initial opinion settings in the SI.

% \begin{figure}[t!]
%      \centering
%      \includegraphics[width=0.45\textwidth]{figures/in_op_SBM_2communities.pdf}
%     \caption{The four initial opinion vectors across different communities}
%     \vspace{-10pt}
%     \label{fig:initial_opinions}
% \end{figure}
% \vspace{-5pt}
%
%We chose the vectors of initial opinions in two ways. Since our main goal is to analyze the impact of the cascade on the polarization, we first chose the initial opinions as the polarizing vectors described in Section~\ref{sec: polarization}. Additionally, we also tested three other initial opinion vectors: $\bm{u}_{unif}$ with entries randomly drawn from a uniform distribution in $[-1,1]$, $\bm{u}_{\sim \mathrm{P}}$ with entries proportional to the Pagerank centrality of nodes, $\bm{u}_{\sim \mathrm{P}^{-1}}$ with entries inversely proportional to the Pagerank centrality of nodes.

\vspace{-5pt}
\subsection{Validation of the approximations}

For each condition, we ran $1{,}000$ simulations of the full sequence of selected cascade sources. For the simulation baseline, we computed each distance independently for every Monte Carlo run and report the mean distance across runs, with 95\% Monte Carlo confidence intervals. We considered the exact diffusion paths as well as the FT and TT approximations (Section~\ref{sec:approximations} where we set threshold $\tau = 0.05$), and compared them with the analytical predictions obtained from Theorem~\ref{cor:fjc_series_cascades}, thus assessing the accuracy of the mean-field approximation and the associated scalable reductions. To quantify distributional differences, we computed three popular distances~\cite{levina2001earth,lin1991divergence}, normalized in $[0,1]$. 
Total Variation captures the overall probability mass mismatch, strongly penalizing discrepancies in peak heights or missing peaks, while ignoring the geometric distance between opinions. Jensen--Shannon measures differences in distributional shape and is sensitive to changes in the number and relative weight of peaks, but does not account for the spatial displacement of probability mass. Wasserstein--1 reflects the minimal cost of transporting probability mass between distributions, capturing shifts in peak locations without over-penalizing small misalignments.

Figure~\ref{fig:validation} shows that the simulated distributions ($sim$ on the x-axis) closely match the theoretical results, with all distances remaining below $0.05$. The Monte Carlo confidence intervals are very narrow, indicating that the observed discrepancies are not driven by simulation variability.  For the FT and TT approximations, TV and JS distances (green and pink bars) can attain substantially larger values, whereas $W_1$ (blue bars) remains below $0.05$ in all cases. This indicates that the overall trend of the distributions is well captured, even when the approximations occasionally mispredict the precise shape of the distribution. A detailed visualization of the full distributions is provided in the SI. FT and TT exhibit very similar accuracy, indicating that pruning low-probability branches has only a limited effect under the considered parameter settings. %Notably, TT occasionally outperforms FT; this occurs because in TT all nodes must be included in the tree, which can lead to an overestimation of opinion shifts for leaf nodes. 
These results empirically validate the analytical recursion of Theorem~\ref{cor:fjc_series_cascades}. 
The small discrepancies observed between simulations and analytical predictions demonstrate that the mean-field decoupling approximation captures the dominant statistical effects of the diffusion–opinion coupling at the expectation level. 
Thus, the MFA provides an accurate and computationally tractable surrogate of the exact stochastic dynamics.

%For each condition, we run 1000 simulations and compare the empirical CDF (Figure~\ref{fig:validation}) of final opinions considering the exact diffusion paths or their approximations FT/TT (Section~\ref{sec:approximations}).  %For each condition, we run 1000 simulations and compare the empirical CDF of final opinions with the theoretical predictions and the FT/TT approximations (Figure~\ref{fig:validation}). 
%
% The close correspondence between theoretical predictions and simulation results across the tested values of $\theta$ provides strong evidence for the robustness and accuracy of the theoretical model. Although both FT and TT approximations deviate from the theoretical curve, they still reproduce the general trend of opinion formation. As anticipated, TT, which relies on graph truncation, yields the least accurate approximation.

 \begin{figure}[t]%
\centering
\subfloat[$\theta = 0.1$]{\includegraphics[width=0.5\textwidth]{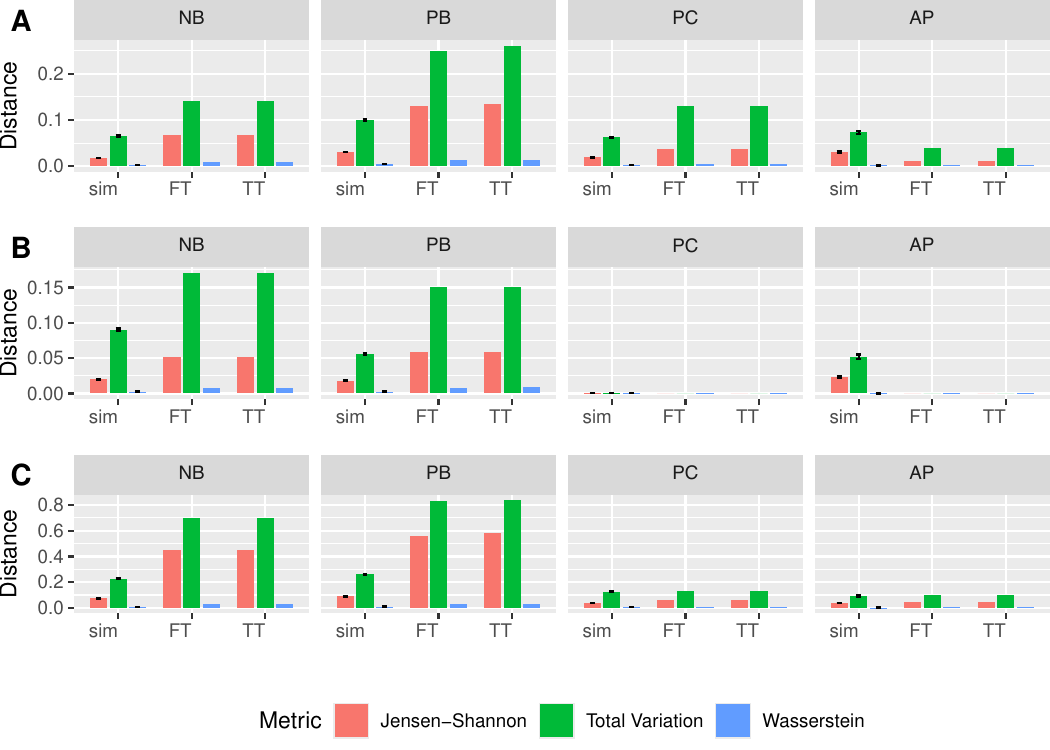}\label{fig:validation_thetaLow}}%
\hfill
\subfloat[$\theta = 0.5$]{\includegraphics[width=0.5\textwidth]{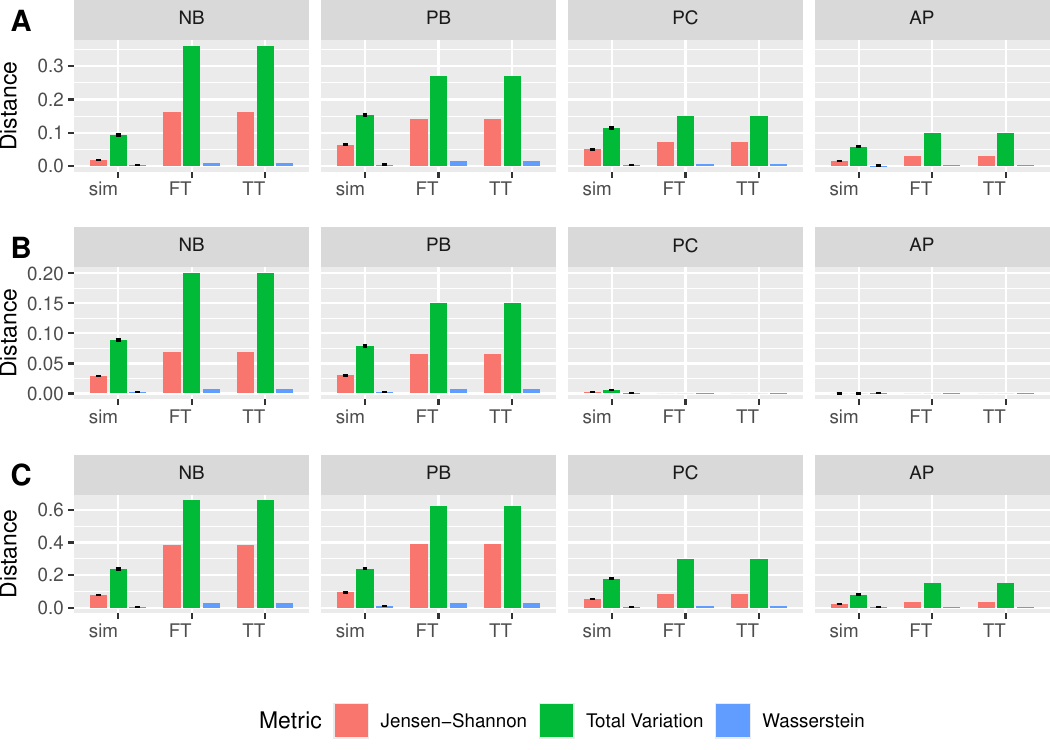}\label{fig:validation_thetaMean}}%
% \hfill
% \subfloat[$\theta = 1$]{\includegraphics[width=0.8\textwidth]{figures/validation_ThetaHighBarabasi_k6_100.pdf}\label{fig:validation_thetaHigh}}%
\caption{
Normalized distributional distances ($[0,1]$) between analytical predictions and simulation/approximation outputs. 
Simulation bars show mean distances over $1{,}000$ Monte Carlo runs with 95\% confidence intervals; FT and TT are deterministic.  (A) $\lambda_i \propto \mathrm{B}_i$; (B) $\lambda_i \propto \mathrm{B}_i^{-1}$; (C) $\lambda_i = 0.6$.
}%\caption{CDF of the final opinions computed numerically using Corollary~\ref{cor:fjc_series_cascades}, and using the approximations (FT and TT) described in Section~\ref{sec:approximations} vs the simulated ones. In the latter case, we plotted the CDF of each node's mean final opinion and shaded the confidence interval obtained with the standard errors. In the different panels, the results for the different choices of susceptibilities: in A they are  directly proportional to $\mathrm{B}_i$, in B they are inversely proportional to $\mathrm{B}_i$ and in C they are equal to $0.6$}
\label{fig:validation}%
\end{figure}
\vspace{-5pt}

\vspace{-3pt}
\subsection{Opinion dynamics and polarization}
%In this section we analyze how opinions evolve in the FJC compared with the classical FJ: thus, we will be able to derive conclusions on the effect of cascade flows over the polarization dynamics. Figures~\ref{fig:results_01} and~\ref{fig:results_05} (corresponding to low and medium diffusion) illustrate the opinion trajectories for each setting of initial opinion: each arrow represents a node, starting from its initial opinion (prejudice) and ending at its final opinion under the FJC model, while black dots indicate the final FJ opinions. %Due to space constraints, we omit $\theta=1$, as the qualitative behavior is analogous. 
%Unless otherwise specified, the following discussion applies to both $\theta$ values, since their main effect lies in modulating the rate rather than the direction of opinion change.  
We now examine how cascade-driven exposure reshapes opinion dynamics relative to the classical FJ model. Rather than analyzing each configuration in isolation, we interpret the results through three recurring mechanisms that differentiate FJC from FJ:
\begin{description}
    \item[Cascade selectivity:] influence propagates only along realized cascade paths, limiting exposure compared to the global averaging of FJ.
    \item[Bridge amplification or suppression:] nodes with high betweenness centrality disproportionately shape FJC dynamics, since cascades amplify their exposure patterns.
    \item[Susceptibility–structure interaction:] the qualitative effect of cascades depends on whether structurally central nodes are more or less susceptible to influence.
\end{description}

These mechanisms arise from a fundamental structural difference between the two models. In the FJ model, influence propagates simultaneously along all social ties, producing symmetric averaging across the network. In contrast, the FJC model transmits influence only along realized cascade paths. As a consequence, interactions between communities occur only when cascades traverse bridge nodes, and such interactions are typically asymmetric, originating from one community and reaching only a subset of the other. 
Figures~\ref{fig:results_01}  and~\ref{fig:results_05}  illustrate these mechanisms under low and moderate diffusion ($\theta = 0.1$ and $\theta = 0.5$).
Each arrow represents a node’s trajectory from its initial opinion to its final opinion under FJC, while black dots denote the corresponding FJ equilibria. Unless otherwise specified, differences across $\theta$ primarily affect the magnitude of shifts rather than their direction.

We now analyze how these mechanisms manifest across the four opinion configurations considered in our experiments.
We begin with the \textsc{pc} configuration. When the two communities start with opposite extreme opinions, the classical FJ model reduces polarization through symmetric averaging: both groups move toward the center and the distance between communities decreases, although distinct clusters may still persist depending on the susceptibility pattern. 
In contrast, FJC dynamics preserve separation between communities. Cascade selectivity limits cross-community exposure and concentrates influence on bridge nodes. Peripheral nodes, deeply embedded within their own community, remain close to their initial positions, while bridge nodes shift most. As a result, inter-community separation persists and complete homogenization does not occur. 
When susceptibility is proportional to centrality (Panel~A), this effect intensifies, as bridge nodes become highly influential and sustain polarization. When susceptibility is inversely proportional to centrality (Panel~B), bridges become stubborn, further reducing cross-community mixing.

We next consider the \textsc{ap} configuration, where one community is neutral and the other extreme. In this regime the two models behave differently but with smaller overall shifts. Under FJ dynamics, the extreme community is partially moderated by the neutral one, producing partial convergence. 
Under FJC, cascade selectivity further weakens cross-community influence, so opinion shifts remain limited and the two communities largely preserve their initial separation. This configuration therefore illustrates how cascade-mediated interactions dampen cross-group influence when the initial opinion distance is moderate.

%The behavior observed in the \textsc{ap} configuration is analogous, but restricted to the opinion domain $[0,1]$. Opinion changes follow the same direction as in the \textsc{pc} case, yet occur at a substantially slower rate. This is primarily due to the reduced opinion distances, which directly scale the magnitude of the interaction forces. In contrast to the \textsc{pc} setting, where communities initially occupy the extreme values $-1$ and $1$, the smaller initial separation in this case results in limited opinion shifts. Consequently, communities exhibit near-stationary behavior. Although the governing dynamics are unchanged, the reduced interaction strength yields an effectively distinct outcome.

In the \textsc{nb} configuration, bridge nodes are initially neutral and connect two extreme communities. Under FJ dynamics, this neutral group acts as a mediator, pulling both communities toward moderation and producing three internally coherent clusters with reduced polarization. 
Under FJC, this mediating role weakens. Because exposure is restricted to realized cascades, neutral bridges do not transmit influence uniformly across the two groups. Peripheral nodes tend to shift more than bridges, and the final separation between communities remains wider than under FJ. When bridge nodes are stubborn (Panel~B), partial blending occurs; when they are highly susceptible (Panel~A), polarization is maintained.

%We next consider the \textsc{nb} configuration, where the two communities begin with opposite opinions and are connected only through a set of neutral bridge nodes. In the FJ model, opinions within each of the three groups—extreme left, extreme right, and neutral—converge internally: the neutral nodes remain near zero, while the two extreme groups move substantially toward neutrality, yielding a relatively narrow final gap. In contrast, the FJC model does not produce such group-level consensus. Notably, the neutral nodes remain close to their initial position, whereas the extreme communities stay significantly farther from neutrality than in the FJ case, resulting in a markedly wider final separation. Please note that, in opposite to \textsc{pc}, opinion shifts occur primarily among peripheral nodes, which are influenced by their direct connections with neutrality.
%
%When susceptibility is directly proportional to centrality, making peripheral nodes more resistant, both models maintain three distinct clusters of opinions (neutral, positive, and negative). Conversely, when susceptibility is inversely proportional to centrality and bridge nodes are more stubborn, FJ shows overlapping clusters partially overlap and more moderate opinions, thereby weakening polarization, as neutral nodes exert a softening influence on the most extreme members of each community. In contrast, in FJC the cascades reduce this influence and maintain communities separated.  

Finally, we analyze the \textsc{pb} configuration, where bridge nodes hold extreme opinions and the remaining nodes are neutral. The outcome depends strongly on the susceptibility of the bridge nodes. Under the classical FJ dynamics, highly susceptible bridges (Panel~A) have their extreme prejudices mitigated through averaging, leading to a weakly segregated or mixed state. If bridges are stubborn (Panel~B), their extreme opinions persist and communities remain separated, while uniform susceptibility (Panel~C) produces an intermediate regime with multiple clusters but reduced distances. 
Under FJC, cascade selectivity again modifies the outcome: influence propagates only along realized cascade paths, so bridges do not necessarily mediate interactions symmetrically across communities. For moderate diffusion intensity ($\theta=0.5$), cascades reaching peripheral nodes may induce partial blending and a more continuous opinion spectrum; when cascade paths are sparser or bridges are effectively stubborn, separation is preserved.

Across all configurations, the key distinction is that FJ enforces global averaging across all edges, whereas FJC restricts influence to cascade-realized exposure paths. This restriction can preserve polarization when cross-community cascades are rare or bridges are stubborn, or produce asymmetric blending when cascades traverse highly susceptible bridges. Thus, cascade-mediated exposure does not simply accelerate opinion change: it qualitatively reshapes which interactions occur and which structural nodes dominate influence.

%Finally, when bridge nodes hold extreme opinions while the remaining nodes start neutral, i.e. \textsc{pb}, the FJ and FJC models display markedly different depolarization behaviors. With equal susceptibility, the FJ model yields a symmetric contraction of the three groups, whereas the FJC model produces partial mixing: extreme bridge nodes are drawn toward moderation, while neutral peripheral nodes shift slightly toward the extremes. For higher values of $\theta$ (Figure~\ref{fig:results_05}), this effect intensifies, with neutral nodes polarizing around $\pm 0.25$, resulting in a  continuous spectrum of opinions. This indicates that when polarization comes from bridge nodes, cascades can promote cross-community blending and  depolarization, depending on the resharing attitude. When bridge nodes are instead stubborn by design (susceptibility inversely proportional to centrality), this blending effect vanishes, and communities remain distinct. Conversely, when bridge nodes are more open-minded (susceptibility directly proportional to centrality), the FJ model generates substantial cross-community mixing, effectively erasing the initial group structure, while the FJC model maintains clearer community boundaries, though with reduced separation.

\vspace{-10pt}
\section{Conclusions}

This paper provided an analytical characterization of the Friedkin--Johnsen on Cascades (FJC) model, explicitly linking stochastic information diffusion with opinion dynamics. We showed that cascade-driven exposure induces a statistical coupling between diffusion realizations and opinion updates, which prevents an exact closed recursion in unconditional expectations.
To obtain a tractable characterization, we introduced a mean-field decoupling approximation (MFA) leading to a linear recursion for the expected opinion trajectory. Simulations validated that this approximation captures the main statistical effects of the diffusion-opinion interaction. We also developed scalable tree-based reductions that substantially reduce computational complexity while preserving key structural properties of cascade-driven influence.
Finally, by comparing FJC with the classical FJ model on SBM networks with two communities and controlled initial opinion configurations, we highlighted a fundamental difference between global averaging and cascade-mediated exposure. Whereas FJ typically promotes homogenization, FJC filters influence through realized cascade paths, enabling regimes of bounded assimilation and persistent polarization depending on community structure and susceptibility patterns. Overall, the proposed framework provides a foundation for studying how cascade-driven interactions reshape collective opinion landscapes in community-structured networks.

\begin{figure*}%[p]
\centering
\text{$\bm{\theta=0.1}$}\\[0.5em]  % <- title above
    %\subfloat[$\theta=0.01$]{
        \includegraphics[width=0.8\textwidth]%{figures/opinions_SBM_2communities_theta01.pdf}\vspace{-5pt}
        {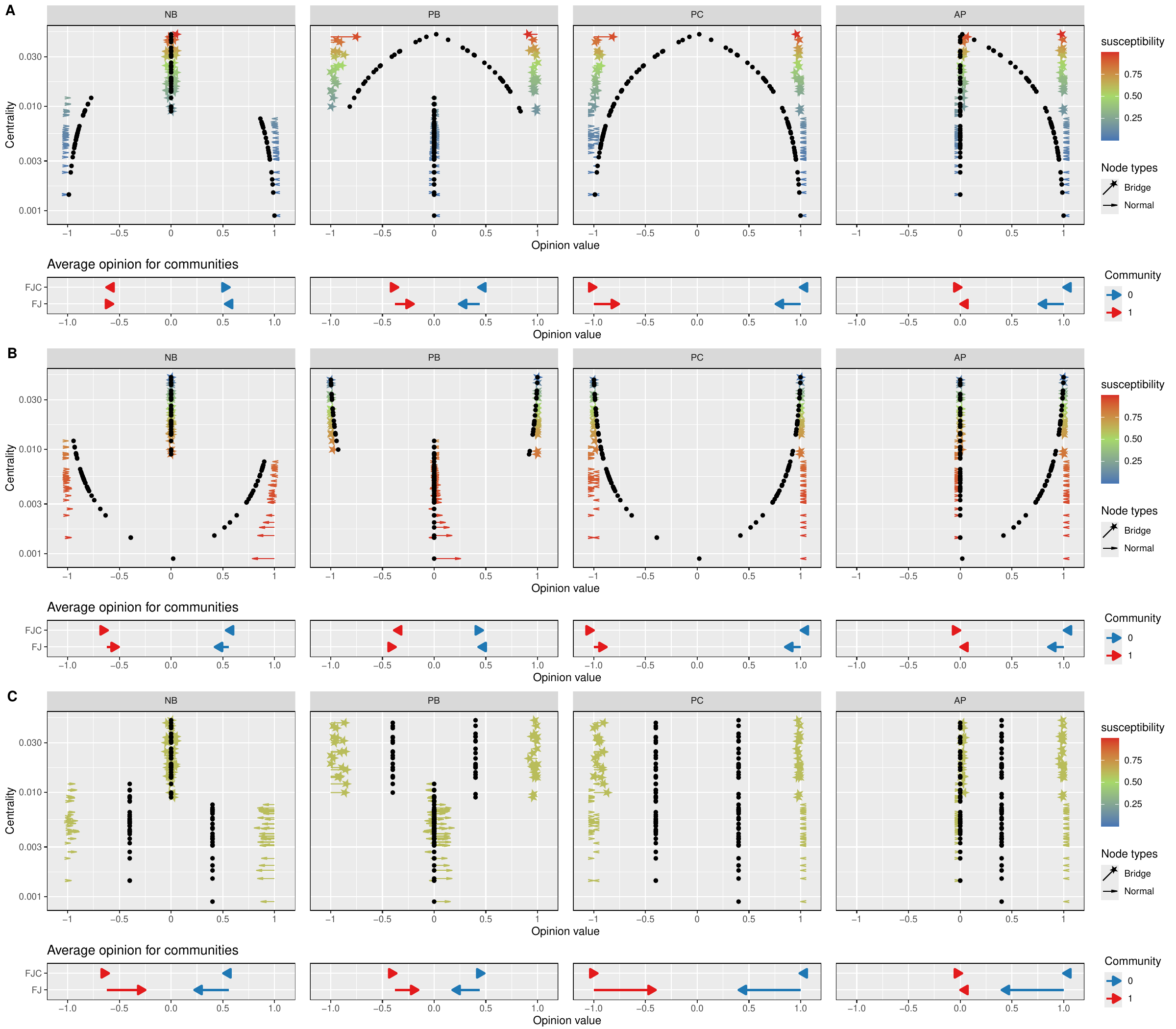}\vspace{-5pt}
    % }\hfill
    % \subfloat[$\theta=0.1$]{
    %     \includegraphics[width=0.9\textwidth]{figures/polVec_theta01_all.pdf}\vspace{-5pt}
    % }\hfill    
    % \subfloat[$\theta=0.5$]{
    %     \centering
    %     \includegraphics[width=0.9\textwidth]{figures/polVec_theta05_all.pdf}\vspace{-5pt}
    % }\hfill
 \caption{Opinions vs. node centrality with $\theta = 0.1$. Each arrow represents a node, starting at its initial opinion (prejudice) and ending at its final opinion under the FJC model. Arrow color indicates node susceptibility; black dots show the final opinions under the FJ model. (A) $\lambda_i \propto \mathrm{B}_i$; (B) $\lambda_i \propto \mathrm{B}_i^{-1}$; (C) $\lambda_i = 0.6$.}
 % \caption{Polarizing opinions vs nodes centrality. Each arrow corresponds to one node in the graph. An arrow starts at the initial prejudice and ends at the final opinion with the FJC model. The color of the arrow corresponds to the susceptibility assigned to the node. The black points represent the final opinion of the corresponding nodes with the FJ model. (A) $\lambda_i \sim \mathrm{P}_i$ (B) $\lambda_i \sim \mathrm{P}_i^{-1}$ (C) $\lambda_i = 0.6$ (i) $\theta =0.01$ (ii) $\theta=0.8$. In each panel, on the left the opinions yielding $P_4$-polarization, on the right the opinions yielding $P_2,P_3$-polarization. With the dotted and dashed line we denote the average initial and final opinion with the FJC, respectively.}
 \label{fig:results_01}
 \vspace{-10pt}
\end{figure*}

\begin{figure*}%[p]
\centering
\text{$\bm{\theta=0.5}$}\\[0.5em]  % <- title above
        \includegraphics[width=0.8\textwidth]%{figures/opinions_SBM_2communities_theta05.pdf}\vspace{-5pt}
        {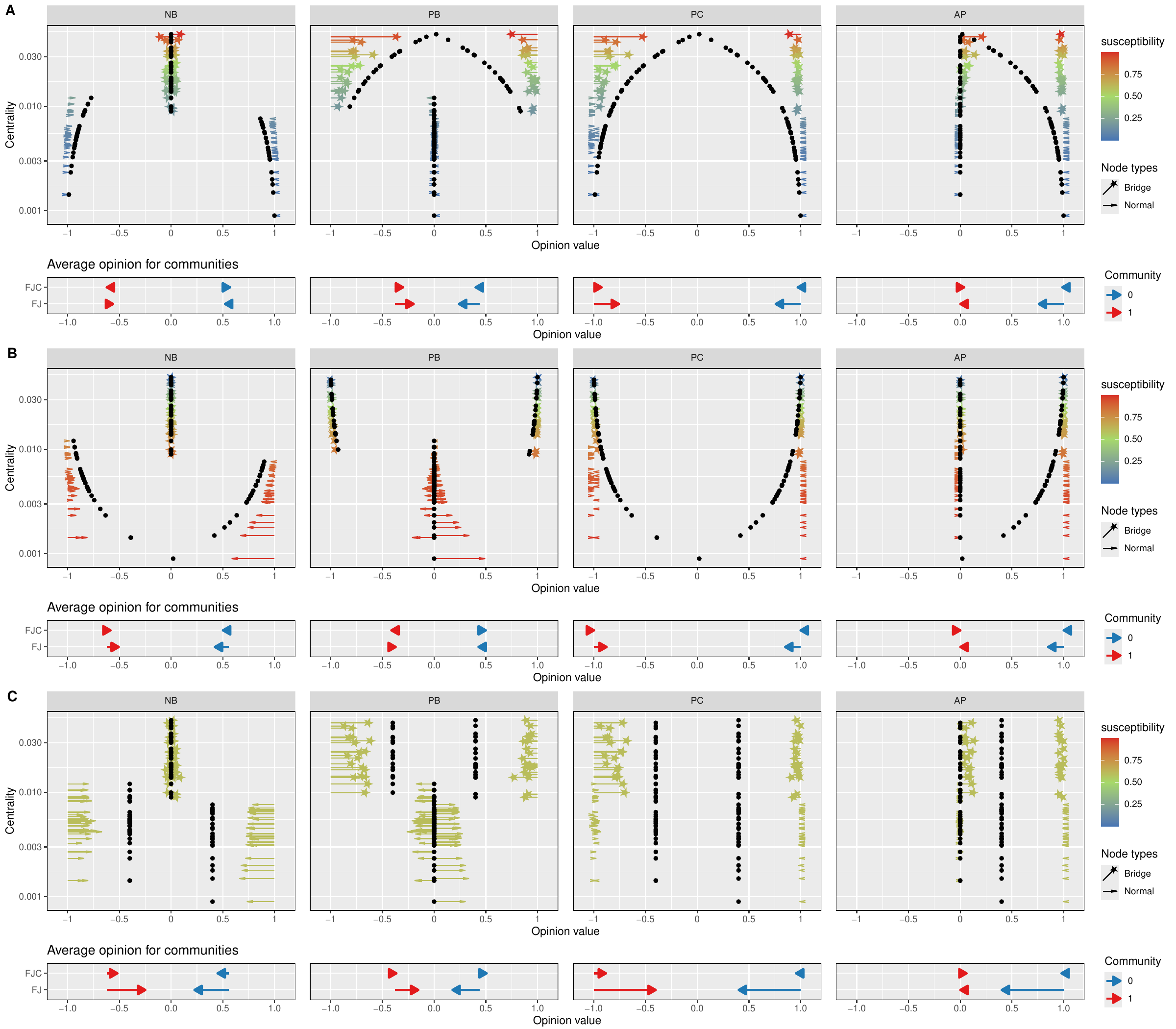}\vspace{-5pt}
 \caption{Opinions vs. node centrality with $\theta = 0.5$. Each arrow represents a node, starting at its initial opinion (prejudice) and ending at its final opinion under the FJC model. Arrow color indicates node susceptibility; black dots show the final opinions under the FJ model.  (A) $\lambda_i \propto \mathrm{B}_i$; (B) $\lambda_i \propto \mathrm{B}_i^{-1}$; (C) $\lambda_i = 0.6$. }
 \label{fig:results_05}
 \vspace{-10pt}
\end{figure*}

\vspace{-5pt}
\bibliography{references.bib}{}
\bibliographystyle{IEEEtran}
%RICORDARSI DI AGGIUNGERE BIO

%\vskip -2.5\baselineskip plus -1fil
\clearpage
%\appendix

%In this section, we will present some of the proofs that were omitted or shortened in the body of the paper.

\end{document}